\documentclass[12pt]{article}
\usepackage{amssymb,amsmath,tikz,amsthm}
\usepackage[shortlabels]{enumitem}
\usepackage{paralist}
\usepackage[all]{xy}
\usepackage{comment}

\usepackage{vmargin}
\setmarginsrb{9mm}{1mm}{9mm}{6mm}%
          {10mm}{7mm}{10mm}{10mm}

\usepackage{hyperref}

\newtheorem{lemma}{Lemma}[section]
\newtheorem{proposition}[lemma]{Proposition}
\newtheorem{theorem}[lemma]{Theorem}
\newtheorem{corollary}[lemma]{Corollary}

\newtheorem{fact}[lemma]{Fact}

\newtheorem{question}[lemma]{Question}

\theoremstyle{definition}
\newtheorem{definition}[lemma]{Definition}
\newtheorem{remark}[lemma]{Remark}
\newtheorem{example}[lemma]{Example}

\def\N{\mathbb N}
\def\R{\mathbb R}
\def\Z{\mathbb Z}

\def\T{\mathbb T}

\def\ent{\mathrm{ent}}
\newcommand\eps{\varepsilon}

\newcommand{\SSS}{\mathbb S}
\newcommand{\TT}{\mathbb T}

\newcommand{\argu}{\hbox to 7truept{\hrulefill}}

\numberwithin{equation}{section}

\newlength{\bibitemsep}
\newlength{\bibparskip}
\let\oldthebibliography\thebibliography
\renewcommand\thebibliography[1]{%
  \oldthebibliography{#1}%
  \setlength{\parskip}{\bibitemsep}%
  \setlength{\itemsep}{\bibparskip}%
}

\author{Ilaria Castellano \and Anna Giordano Bruno \and  Nicol\`o Zava}

\newcommand{\Addresses}{{
  \bigskip
  \footnotesize

 \noindent I.~Castellano, \textsc{Fakult\"at f\"ur Mathematik, Universit\"at Bielefeld, D-33501 Bielefeld, Germany}\\
  \textit{E-mail address}: \texttt{ilaria.castellano@math.uni-bielefeld.de}

  \medskip

 \noindent  A.~Giordano Bruno, \textsc{Dipartimento di Scienze Matematiche,  Informatiche e Fisiche, Universit\`a degli Studi di Udine, 33100 Udine, Italy}\\
  \textit{E-mail address}: \texttt{anna.giordanobruno@uniud.it}

  \medskip

 \noindent N.~Zava, \textsc{Institute of Science and Technology Austria (ISTA), 3400 Klosterneuburg, Austria}\\
  \textit{E-mail address}: \texttt{nicolo.zava@gmail.com}
}}

\title{Weakly weighted generalised quasi-metric spaces and semilattices}

\date{}

\begin{document}

\maketitle

\abstract{Motivated by recent applications to entropy theory in dynamical systems, we generalise notions introduced by Matthews and define 
weakly weighted and componentwisely weakly weighted (generalised) quasi-metrics. We then systematise and extend to full generality the correspondences between these objects and other structures arising in theoretical computer science and dynamics. In particular, we study the correspondences with weak partial metrics, and, if the underlying space is a semilattice, with invariant (generalised) quasi-metrics satisfying the descending path condition, and with strictly monotone semi(-co-)valuations.

We conclude discussing, for endomorphisms of generalised quasi-metric semilattices, a generalisation of both the known intrinsic semilattice entropy and the semigroup entropy.}

\bigskip

                               

\section{Introduction}

The notion of metric space has been generalised in several directions with the aim of producing non-Hausdorff topologies. 
For instance, Scott's breakthrough works \cite{Sco1,Sco2} used non-Hausdorff topologies to describe partial objects in computation (see also the survey~\cite{AbrJun} and the monograph~\cite{Gou}). 
Moreover, partial metrics were introduced by Matthews~(\cite{Mat}) and allow the distance from a point to itself to be non-zero.
On the other hand, quasi-metrics abandon the symmetry axiom; they were introduced and studied for the first time by Wilson~(\cite{Wil}), but this notion can be traced back to Hausdorff~(\cite{Hau}). 
 
Matthews introduced the notion of a \emph{weighted quasi-metric}, that is a quasi-metric $d$ on a non-empty set $X$ such that there exists a function $w\colon X \to\R_{\geq0}$ with 
\begin{equation}\label{weqintro}
d(x,y)+w(x)=d(y,x)+w(y)\quad\text{for every $x,y\in X$}.
\end{equation}
Indeed, he showed a precise correspondence between partial metrics and weighted quasi-metrics, where the specialisation orders of the partial metric and the corresponding weighted quasi-metric coincide:
\begin{equation}\label{Meq}
\xymatrix@1@C=35pt@M=1.25ex{
 \text{partial metrics} \ar@{<->}[r]^-{\cite{Mat}} & \text{weighted quasi-metrics}.}
\end{equation}
As immediately noticed by Matthews, not every quasi-metric is weighted, and moreover the question he posed of which quasi-metrizable spaces admit weighted quasi-metrics is still open -- see the paper~\cite{KV94} by K\"unzi and Vajner on topological aspects of weighted quasi-metrics.

\smallskip
Schellekens~(\cite{Sch}) noticed that several examples of partial metric spaces arising  in quantitative domain theory share a further underlying structure: 
the specialisation order associated to the weighted quasi-metric (equivalently, to the partial metric) turns them into semilattices. In this case we speak about \emph{quasi-metric semilattices} with a slight difference with the terminology used in~\cite{Sch}, which adopts the classical one for quasi-metric lattices from~\cite{Weber1,Weber2} (see Definition~\ref{invdef} and the comment after it).  

Inspired by the classical connection between (strictly increasing) valuations and (pseudo-)metrics on a lattice (see Birkhoff's monograph~\cite{Bir}), a correspondence between partial metrics and valuations on semilattices was studied by O'Neill~(\cite{O96}) and further discussed in~\cite{BukScott,BukShorina}.

As previously done by Nakamura~(\cite{Nak70}), Schellekens~(\cite{Sch}) introduced semi(-co-)valuations on semilattices, showing how these notions precisely extend the classical concept of valuations for lattices. Moreover, he provided a correspondence between \emph{invariant} (see Definition~\ref{invdef}) (co-)weighted quasi-metrics and strictly increasing semi(-co-)valuations~(\cite{Sch}), and a characterisation of invariant quasi-metrics on semilattices as those satisfying the \emph{descending path condition} (briefly, {\rm (DPC)} -- see Definition~\ref{def:DPC}) and the extension property~(\cite{Sch_1}): 

\begin{equation}\label{eq:correspondence_semilattices}
\xymatrix@1@C=10pt@M=1.25ex{\txt{invariant\\ partial metrics}\ar@{<->}[rr]^-{\cite{Mat}} && \txt{invariant weighted\\ quasi-metrics}\ar@{<->}[rr]^-{\cite{Sch_1}}\ar@{<->}[d]^-{\cite{Sch}} &&\txt{invariant quasi-metrics with\\ (DPC) and extension property}\\ && \txt{strictly monotone\\positive semi(-co-)valuations.} &&
}\end{equation}

\smallskip
%
The main aim of this paper is to complete, extend and systematise 
the above relations by Matthews and Schellekens. In order to achieve this, inspired by the notion of \emph{weak partial metrics} due to O'Neill~(\cite{O96}, 
see Definition~\ref{def:w_partial_m}), we introduce \emph{weakly weighted quasi-metrics} $d$ on a space $X$ by allowing the function $w\colon X\to \R$, which satisfies \eqref{weqintro},  to take any real value (see Definition~\ref{def:w_weighted}).  
%
%
%
This new notion immediately allows us to extend the correspondence in \eqref{Meq} to weak partial metrics and weakly weighted quasi-metrics: 
\begin{equation}\label{p-wintro}
\xymatrix@1@C=20pt@M=1.25ex{
\text{weak partial metrics} \ar@{<->}[rr]^-{\txt{Cor.~\ref{coro:wwqm_vs_wpm_one_to_one}}}& & \text{weakly weighted quasi-metrics.}}
\end{equation}
Then, in the spirit of Schellekens' work, we extend his results for semilattices collected in \eqref{eq:correspondence_semilattices}. In particular, 
we characterise the invariant quasi-metric semilattices that are weakly weighted  as those satisfying 
{\rm (DPC)}, 
and 
prove the existence of a correspondence between weak weights for invariant quasi-metrics and strictly monotone semi(-co-)valuations 
for a semilattice
: 
\begin{equation}\label{eq:correspondences_semilattices_new}\xymatrix@1@C=25pt@M=1.25ex{
\txt{invariant weak\\partial metrics}\ar@{<->}[rr]^-{\txt{Rem.~\ref{rem:specialisation_order_by_p}}} && \txt{invariant weakly \\ weighted quasi-metrics} \ar@{<->}[d]^-{\txt{Cor.~\ref{coro:w_wei_to_valu2}}} && \txt{invariant quasi-metrics\\ with {\rm(DPC)}} \ar@{<->}[ll]_-{\txt{Cor.~\ref{theo:oc_vs_weak_wei}}} \\
&&\txt{strictly monotone\\ semi(-co-)valuations.}  &&
}\end{equation}
Furthermore, our extensions allow for alternative proofs of the classical correspondences represented in \eqref{eq:correspondence_semilattices}.

\smallskip
A quasi-metric $d$ on a space $X$ is a \emph{generalised quasi-metric} if it can also assume the value $\infty$.
Generalised quasi-metric semilattices play a central role in the theory of intrinsic entropy. 
Indeed, the first attempt at a lattice-theoretic approach to entropy is due to Nakamura, who proved in \cite{Nak70} that Shannon's entropy from information theory (\cite{Sha}) can be characterised as a semivaluation on a semilattice.
Then, a notion of \emph{semigroup entropy} for normed semigroups and their monotone endomorphisms was proposed in~\cite{uatc} in order to find a general scheme for (almost) all classical entropies in mathematics; in most of the specific cases the underlying semigroup is a semilattice. 
Finally, in~\cite{CDFGBT} a notion of {\em intrinsic semilattice entropy} was introduced in the category of generalised quasi-metric semilattices and non-expansive homomorphisms as the key to a unifying approach able to cover all (or, at least, most) of the intrinsic-like entropies in literature. 
Since all of the examples of generalised quasi-metric semilattices appearing in~\cite{CDFGBT} satisfy the property {\rm (DPC)}, it is natural to ask whether we could retrieve a scheme of correspondences similar to~\eqref{eq:correspondences_semilattices_new} in this more general context.

To this end, 
first we introduce \emph{generalised weak partial metrics} and 
further extend the correspondence in~\eqref{p-wintro} to 
generalised weak partial metrics and weakly weighted generalised quasi-metrics: 
\[\xymatrix@1@C=20pt@M=1.25ex{
\text{generalised weak partial metrics} \ar@{<->}[rr]^-{\txt{Cor.~\ref{cor:pm ww}}} &&\text{weakly weighted generalised quasi-metrics.}
	}\]

We then investigate possible extensions of \eqref{eq:correspondences_semilattices_new} to the generalised setting. It is clear that weakly weighted generalised quasi-metrics and generalised weak partial metrics are not the right tool to characterise generalised quasi-metric semilattices with (DPC). In fact, a weakly weighted generalised quasi-metric space is actually a quasi-metric space provided that its specialisation order induces a semilattice (see Remark~\ref{rem:gqm_semil_not_weighted}). A more suitable alternative descends from the following observation. 	
The absorbing element $\infty$ forces the generalised quasi-metric space $(X,d)$ to split in connected components $\mathcal Q(x)$, with $x\in X$, and so it induces an equivalence relation $\cong_d$ on $X$. 
Each $\mathcal Q(x)$ is a quasi-metric space with the quasi-metric $d\restriction_{\mathcal Q(x)}$ given by the restriction of $d$. 
We then 
call $(X,d)$ \emph{componentwise weakly weighted} when there exists a function $w\colon X\to \R$ such that, for every $x\in X$, $w\restriction_{\mathcal Q(x)}$ is a weak weight for 
$d\restriction_{\mathcal Q(x)}$ 
(see Definition~\ref{def:componentwisely_ww_gen_qm}).

Moreover, we introduce the notion 
of \emph{generalised semi(-co-)valuation} (see Definition~\ref{def:gen_sv_and_scv}) for generalised quasi-metric semilattices, and extend the correspondences in~\eqref{eq:correspondences_semilattices_new}
to the most general framework:
$$\xymatrix@1@C=20pt@M=1.25ex{
\txt{invariant generalised\\ quasi-metrics  \\ satisfying {\rm(DPC)}} \ar@{<->}[rr]^-{\txt{Th.~\ref{theo:gen_oc_vs_weak_wei}}}&& \txt{invariant componentwise\\ weakly weighted \\ generalised quasi-metrics} \ar@{<->}[rr]^-{\txt{Cor.~\ref{coro:w_wei_to_valu2gen}}} && \txt{strictly monotone\\ generalised \\ semi(-co-)valuations.}
}$$

\smallskip
At the end of this paper, as an application of the above listed results we introduce a new notion of entropy. For an endomorphism $\phi$ of a semilattice $X$ with an equivalence relation $\cong$, first we introduce the concept of $(\phi,\cong)$-inert element, which extends that of $\phi$-inert element for a non-expansive endomorphism of a generalised quasi-metric semilattice introduced in~\cite{CDFGBT}. 
Then, using $(\phi,\cong)$-inert elements, we propose a new intrinsic entropy for $\phi$ with respect to $\cong$.
This entropy generalises both the semigroup entropy for normed semilattices from~\cite{uatc} and the intrinsic semilattice entropy from~\cite{CDFGBT}.

\medskip
The paper is organised as follows.

Section~\ref{wwspacessec} collects the main definitions, some basic properties of (componentwisely) weakly weighted (generalised) quasi-metric spaces and provides several examples. Those coming from entropy theory will be exploited  throughout the entire paper. 
In particular, \S\ref{ss:wwqm} treats the weakly weighted quasi-metric spaces and the comparison to the standard notions of (co-)weighted quasi-metric spaces, \S\ref{ss:ww} concerns weakly weighted generalised quasi-metric spaces, while \S\ref{ss:ww_gen} the even more general notion of componentwise weakly weighted generalised quasi-metric spaces.

Section~\ref{pmsec} is dedicated to partial metrics. In \S\ref{pmwwsec} we recall the notion of a weak partial metric 
and provide the above mentioned correspondence between weakly weighted quasi-metrics and weak partial metrics. 
In \S\ref{gpmwwsec} this correspondence is extended to weakly weighted generalised quasi-metrics and generalised weak partial metrics.

In Section~\ref{sec:quasi-metric_semilattices}, we turn our attention to generalised quasi-metric semilattices: in \S\ref{specordersec} we analyse the properties of the specialization order and 
{\rm (DPC)}, while \S\ref{invariantsec} concerns the invariance property. 

Section~\ref{wwqmsec} is dedicated to the characterization of weakly weighted invariant quasi-metric semilattices as those that satisfy {\rm (DPC)} (\S\ref{ss:weighting}) and (in \S\ref{ss:semivaluation}) to the correspondence between weakly weighted invariant quasi-metrics and strictly 
monotone semi(-co-)valuations.

In Section~\ref{ss:wwgen}, we extend the results from the previous section to generalised quasi-metric semilattices, by defining suitable notions of generalised semi(-co-)valuations.

In Section~\ref{entropysec}, we introduce the new intrinsic entropy for an endomorphism $\phi$ of a meet-semilattice $X$ with an equivalence relation $\cong$, and we recover the semigroup entropy (\S\ref{hSesec}) and the intrinsic semilattice entropy (\S\ref{tildehsec}) as particular cases.

\medskip
\noindent {\bf Notation.} We denote by $\R$ and $\N$ the real and natural numbers (including $0$), respectively. Moreover, we set, for a real constant $c$, $\R_{\geq c}=\{x\in\R\mid x\geq c\}$ and $\R_{>c}=\{x\in\R\mid x>c\}$. Finally, $\N_+=\N\setminus\{0\}$.

\section{Weakly weighted (generalised) quasi-metric spaces}\label{wwspacessec}

\begin{definition}\label{gqmdef}
A  {\em generalised quasi-metric} on a non-empty set $X$ is a function $d\colon X \times X \to\R_{\geq 0}\cup\{\infty\}$ satisfying the  properties:
\begin{compactenum}
	\item[(QM1)] \label{qm1}  for $x,y\in X$, $d(x,y)=d(y,x)=0$ if and only if $x=y$;
	\item[(QM2)] \label{qm2}  $d(x,z)\leq d(x,y)+d(y,z)$, for every $x,y,z\in X$, with   $r<r + \infty = \infty + \infty = \infty$ for every $r\in\R_{\geq0}$.
\end{compactenum}
Whenever  the function $d$ admits only finite values, i.e., $d\colon X \times X \to\R_{\geq0}$, we 
simply call it 
\emph{quasi-metric}. 

The pair $(X,d)$ is a ({\em generalised}) {\em quasi-metric space} provided that $d$ is a (generalised) quasi-metric on $X$.
\end{definition}

A (generalised) quasi-metric $d$ on a set $X$ is said to be a ({\em generalised}) {\em metric} if it satisfies the following axiom:
\begin{compactenum}
	\item[(QM3)] $d$ is {\em symmetric} (i.e., $d(x,y)=d(y,x)$ for every $x,y\in X$).
\end{compactenum}
In this case, the pair $(X,d)$ is a ({\em generalised}) {\em metric space}.

Given a (generalised) quasi-metric $d$ on a set $X$, we define the {\em conjugated (generalised) quasi-metric $d^{-1}$} by 
\begin{equation}\label{eq:conj_qm}
	d^{-1}(x,y):=d(y,x),\quad\text{for every $x,y\in X$,}
\end{equation}
and the {\em symmetrisation} $d^s$ of $d$, which is a (generalised) metric, by the law
\begin{equation}
	d^s(x,y):=\max\{d(x,y),d^{-1}(x,y)\}=\max\{d(x,y),d(y,x)\},\quad\text{for every $x,y\in X$.}
\end{equation}

If $(X,d)$ is a generalised quasi-metric space and $Y$ is a subset of $X$, then $Y$ becomes a generalised quasi-metric space with the restriction of $d$ to $Y$, that we denote by $d\restriction_Y$ with some abuse of notation.


\begin{definition}\label{Qdef}
	Let $(X,d)$ be a generalised quasi-metric space. 
	We define 
	the equivalence relation $\cong_d$ on $X$ as follows:
	$$x\cong_dy\ \text{if}\ d^s(x,y)<\infty, \quad \text{for every $x,y\in X$.}$$
	For every $x\in X$, the {\em connected component of $x$} is the equivalence class $[x]_{\cong_d}$ of $x$, which we denote by $\mathcal Q_X(x)$ for consistency with \cite{Za_q_c}.
\end{definition}
If there is no risk of ambiguity, we simply refer to $\mathcal Q_X(x)$ as $\mathcal Q(x)$. 
The family of connected components provides a partition of $X$ with the following further property: for every $x\in X$, $(\mathcal Q(x),d\restriction_{\mathcal Q(x)})$ is a quasi-metric space. 

Notice that quasi-metric spaces can be regarded as  generalised quasi-metric spaces with only one connected component.

 A map $f\colon(X,d_X)\to(Y,d_Y)$ between (generalised) quasi-metric spaces is {\em non-expansive} (or {\em $1$-Lipschitz}) if, for every $x,y\in X$, $d_Y(f(x),f(y))\leq d_X(x,y)$.

\subsection{Weakly weighted quasi-metric spaces}\label{ss:wwqm}

We start with one of the central notions of the paper. 

\begin{definition}\label{def:w_weighted}
	Let $(X,d)$ be a generalised quasi-metric space. Then $(X,d)$ and $d$ are called {\em weakly weighted} if there exists a  function $w\colon X\to\R$, called {\em weak weight} for $d$, such that, for every $x,y\in X$,
	\begin{equation}\label{eq:w_weight}
		d(x,y)+w(x)=d(y,x)+w(y).
	\end{equation}
	We also say that $d$ is weakly weighted by $w$.
\end{definition}

Clearly, every generalised metric is weakly weighted. In fact, \eqref{eq:w_weight} is fulfilled if we take $w$ to be the constant function $0$ (or, in general, any constant function). Moreover, a generalised quasi-metric is weakly weighted by a function $w$ if and only if its conjugate is weakly weighted by $-w$.


The definition above is inspired by the following concepts of Matthews~(\cite{Mat}) and Schellekens~(\cite{Sch}), respectively. A quasi-metric space $(X,d)$ and its quasi-metric $d$ are said to be:
\begin{compactenum}[-]
	\item {\em weighted} if there exists a  function $w\colon X\to\R_{\geq 0}$, called {\em weight}, satisfying \eqref{eq:w_weight} for every pair $x,y\in X$;
	\item {\em co-weighted} if there exists a  function $w\colon X\to\R_{\geq 0}$, called {\em co-weight}, such that, for every $x,y\in X$,
	\begin{equation}\label{eq:co-weight}
		d(x,y)+w(y)=d(y,x)+w(x);\end{equation}
	equivalently, the conjugated quasi-metric $d^{-1}$ is weighted by $w$.
\end{compactenum}

\smallskip
In this subsection, we first focus on weakly weighted quasi-metric spaces, motivated by the connection with weighted and co-weighted quasi-metrics that we formalise in Remark~\ref{rem:weakly_weighted_qm}. We leave to the subsequent subsection considerations regarding weakly weighted generalised quasi-metric spaces.

\begin{remark}\label{rem:weakly_weighted_qm}
	Let $(X,d)$ be a quasi-metric space.
	If $w\colon X\to\R_{\geq 0}$ is a weight for $d$, then trivially $w$ is a weak weight for $d$, and thus $(X,d)$ is weakly weighted. 
	Similarly, if $w\colon X\to\R_{\geq 0}$ is a co-weight for $d$, then $-w$ satisfies \eqref{eq:w_weight} and so it is a weak weight for $(X,d)$. 
	
	Hence, Definition~\ref{def:w_weighted} generalises both weighted and co-weighted quasi-metric spaces.  
	We will see in Example~\ref{ex:w_weighted_qm}(c) that there are weakly weighted quasi-metric spaces that are neither weighted nor co-weighted.
\end{remark}

\begin{proposition}\label{prop:subspace_wei}
	Let $(X,d)$ be a (weakly/co-) weighted quasi-metric space, $w\colon X\to\R$ be a (weak/co-) weight, and $Y\subseteq X$. Then $(Y,d\restriction_{Y})$ is (weakly/co-) weighted and $w\restriction_{Y}$ is a (weak/co-) weight.
\end{proposition}

Weak weights, weights and co-weights are not unique, as stated in the following proposition. On the other hand, the result 
shows that two weak weights  for the same quasi-metric 
differ only by a constant. For the sake of brevity, if $S$ is a 
non-empty set, $f\colon S\to\R$ a function and $c\in\R$, we denote by $f+c$ the function from $S$ to $\R$ defined by the law $(f+c)(x)=f(x)+c$, for every $x\in S$.
Moreover, we write $f\geq 0$ ($f\leq 0$) if, for every $x\in S$, $f(x)\geq 0$ ($f(x)\leq 0$, respectively).

\begin{proposition}\label{prop:w+c}
	Let $(X,d)$ be a quasi-metric space and 
	$w,w^\prime\colon X\to\R$.
	\begin{compactenum}[(a)]
		\item Suppose that $w$ is a weak weight for $d$. 
		Then $w^\prime$ is a weak weight for $d$ if and only if $w^\prime=w+c$ for some $c\in\R$. 
		\item {\rm (\cite[Proposition 8]{Sch})} Suppose that $w$ is a (co-)weight for $d$. 
		Then $w^\prime$ is a (co-)weight for $d$ if and only if $w^\prime=w+c$ for some $c\in\R$ and $w^\prime\geq 0$.
	\end{compactenum}
\end{proposition}
\begin{proof}
	Let us prove item (a), while item (b) can be similarly derived. Suppose then that $w$ is a weak weight for $d$. It is easy to see that for every $c\in\R$ the function $w+c$ still satisfies \eqref{eq:w_weight}. 
	
	Suppose now that $w^\prime$ is another weak weight for $d$. Then, for every $x,y\in X$,
	$$\begin{aligned}w(x)-w^\prime(x)&\,=d(x,y)+w(x)-(d(x,y)+w^\prime(x))=\\
		&\,=d(y,x)+w(y)-(d(y,x)+w^\prime(y))=w(y)-w^\prime(y),\end{aligned}$$
	which is obtained using \eqref{eq:w_weight} for both $w$ and $w^\prime$. Hence $w$ and $w^\prime$ differ by a constant, and so the claim follows.
\end{proof}

Let us introduce an equivalence relation $\sim$ on $\R^X$, which is the family of all functions from a non-empty set $X$ to $\R$, by putting, 
\begin{equation}\label{sim} 
	\text{for every $f,g\in\R^X$, $f\sim g$ if there exists a constant $c\in\R$ such that $g=f+c$.} 
\end{equation}
Then Proposition~\ref{prop:w+c}(a) and (b) can be rewritten respectively in the following way: 
\begin{compactenum}[(a)]
	\item if $w$ is a weak weight for $d$, then $[w]_\sim$ is the family of all weak weights for $d$;
	\item if $w$ is a (co-)weight for $d$, then $[w]_\sim\cap\R_{\geq 0}^X$ is the family of all (co-)weights for $d$.
\end{compactenum}


Using Proposition~\ref{prop:w+c}, we can give a useful characterisation of those weakly weighted quasi-metric spaces that are (co-)weighted.

\begin{proposition}\label{prop:bound}
	For  a weakly weighted quasi-metric space $(X,d)$  the following equivalences hold:
	\begin{compactenum}[(a)]
		\item $(X,d)$ is weighted if and only if some (equivalently, every) weak weight $w\colon X\to \R$ for $d$ has a lower bound $\underline{b}_w$, i.e., $w(X)\subseteq\, [\underline{b}_w,+\infty[$.
		
		In particular, given a weak weight $w$,  the weak weight $w-c\in[w]_{\sim}$ is a weight for every $c\leq \underline{b}_w$;
		\item $(X,d)$ is co-weighted if and only if some (equivalently, every) weak weight $w\colon X\to \R$ for $d$ has an upper bound $\bar{b}_w$, i.e., $w(X)\subseteq\ ]-\infty, \bar{b}_w]$.
		
		In particular, given a weak weight $w$ for $d$,  the weak weight $-w+c\in[-w]_{\sim}$ is a co-weight for $d$ for every $c\geq \bar{b}_w$.
	\end{compactenum}
\end{proposition}
\begin{proof}
	Let us prove item (a), while the dual item (b) can be similarly shown. 
	
	If $(X,d)$ is weighted, then there is a weak weight $w$ 
	for $d$ with non-negative values (see Proposition~\ref{prop:w+c}(b))  and so $\underline{b}_w=0$ is a lower bound.  By Proposition~\ref{prop:w+c}(a),   if a particular weak weight 
	for $d$ is bounded from below, then  
	every weak weight for $d$ has a lower bound.
	
	Conversely, suppose that there is a weak weight $w$ for $d$ bounded from below by a constant $\underline{b}_w$. If $\underline{b}_w\geq 0$, then $w$ is a weight for $d$ by definition and $(X,d)$ is weighted. Otherwise, we can use $w-c\in[w]_{\sim}$ to weight $(X,d)$ whenever $c\leq \underline{b}_w$.
\end{proof}

\begin{remark}\label{fading}
	Let $(X,d)$ be a quasi-metric space.
	In \cite{Sch}, a function $f\colon X\to\R$ such that $\inf_{x\in X}f(x)=0$ is called {\em fading} and \cite[Proposition 8]{Sch} (see Proposition~\ref{prop:w+c}(b)) shows that, if $w$ is a (co-)weight for $d$, then there exists a unique fading (co-)weight for $d$. 
	Therefore, we can say that a weakly weighted quasi-metric space $(X,d)$ is weighted if and only if the class of all its weak weights $[w]_{\sim}$ contains a fading element. In particular, the fading representative can be explicitly obtained as follows: if $w$ is a weak weight for $d$, the function $w^\prime=w-\inf_{x\in X}w(x)$ has the desired property. 
	
A dual observation holds for co-weights.
\end{remark}

Before presenting some examples, we state a technical property of weakly weighted quasi-metric spaces that will be useful in the sequel:

\begin{lemma}
	\label{lem:ww property}
	Let $(X,d)$ be a generalised quasi-metric space that is weakly weighted by $w$. For every $x,y,z\in X$ such that $d(x,z)=d(x,y)+d(y,z)$, one has $d(z,x)=d(z,y)+d(y,x)$.
\end{lemma}
\begin{proof} 
	Let $x,y,z\in X$ be such that $d(x,z)=d(x,y)+d(y,z)$. Then \eqref{eq:w_weight} yields
	\begin{align*}
		d(z,x)&\,=d(z,x)+w(z)-w(z)=d(x,z)+w(x)-w(z)=\\
		&\,=d(x,y)+d(y,z)+w(x)-w(z)=d(y,x)+w(y)+d(y,z)-w(z)=\\
		&\,=d(y,x)+d(z,y)+w(z)-w(z)=d(y,x)+d(z,y). \qedhere
	\end{align*}
\end{proof}

Below we provide several examples of weak weights. In particular, Example~\ref{ex:w_weighted_qm}(c) shows that Proposition~\ref{prop:bound} makes it is easier to produce  
weakly weighted quasi-metric spaces that are neither weighted nor co-weighted. Moreover, it exposes a method to produce (co-)weights: it suffices to bound the weak weight involved by restricting the domains.


\begin{example}\label{ex:w_weighted_qm}
	\begin{compactenum}[(a)]
		\item Let $\SSS=\{0,1\}$ be a two-point space endowed with the quasi-metric $d_{\SSS}$ defined as follows: $d_{\SSS}(0,1)=0$ and $d_{\SSS}(1,0)=1$. The topological space $(\SSS,d_{\SSS})$ is also known as {\em Sierpi\'nski space}. Then $d_{\SSS}$ is both weighted and co-weighted. Take, for example, the weight $w(0)=1$ and $w(1)=0$ and the co-weight $w^\prime(0)=0$ and $w^\prime(1)=1$.
		\item On the three point space $\TT=\{0,1,2\}$, define the quasi-metric $d_\TT$ as follows: $$d_\TT(0,1)=d_\TT(0,2)=d_\TT(1,2)=0,\text{ }d_\TT(2,1)=d_\TT(1,0)=1\text{ and }d_\TT(2,0)=2.$$
		Then $d_\TT$ is both weighted and co-weighted. We can indeed define a weight $w(0)=2$, $w(1)=1$ and $w(2)=0$, and a co-weight $w^\prime(0)=0$, $w^\prime(1)=1$ and $w^\prime(2)=2$.
		\item Let us endow $\R$ with the quasi-metric $d$ defined as follows:
		$$d_\R(x,y)=\max\{x-y,0\}.$$
		Then $d_\R$ is weakly weighted by the weak weight $w$ defined by $w(x)=-x$, for every $x\in X$. By virtue of Proposition~\ref{prop:bound}, since $w$ is neither lower nor upper bounded, $(\R,d_\R)$ is neither weighted nor co-weighted. Again Proposition~\ref{prop:bound} implies that $(\R_{\geq 0},d_\R\restriction_{\R_{\geq 0}})$ is co-weighted and $(\R_{\leq 0},d_\R\restriction_{\R_{\leq 0}})$ is weighted.
	\end{compactenum}
	Notice that  $d_\SSS$ and $d_\TT$ coincide with the quasi-metric induced by $d_\R$ on the subsets $\{0,1\}$ and $\{0,1,2\}$, respectively.
\end{example}

Next we see that not all quasi-metrics are weakly weighted. It should be noticed the difference between the quasi-metric $d_\TT$ in Example~\ref{ex:w_weighted_qm}(b) and the ones defined in Example~\ref{ex:qm_not_weighted1} and Example~\ref{ex:qm_not_weighted2} on the same three-point set.

\begin{example}\label{ex:qm_not_weighted1}
	On a three-point space $\TT=\{0,1,2\}$, define a quasi-metric $d$ as follows:
	$$d(0,1)=d(0,2)=d(1,2)=0,\text{ and }d(2,1)=d(2,0)=d(1,0)=1.$$
	Then $d$ is not weakly weighted. Suppose by contradiction that a weak weight $w$ 
	for $d$ exists. Since $d(0,2)=d(0,1)$ and $d(2,0)=d(1,0)$, \eqref{eq:w_weight} implies that $w(0)=w(1)$, which is a contradiction because of
	$$w(0)=d(0,1)+w(0)=d(1,0)+w(1)=1+w(1).$$
\end{example}

The following examples come from entropy theory and will be exploited  throughout the entire paper.

\begin{example}\label{entex}
	\begin{compactenum}[(a)]
		\item Let $S$ be a finite set. On the power set $\mathcal P(S)$, we define the quasi-metric $d_{\mathcal P(S)}$ as follows:
		\begin{equation}\label{eq:d_P(X)}
			d_{\mathcal P(S)}(A,B):=\lvert B\setminus A\rvert=\lvert(A\cup B)\setminus A\rvert, \quad\text{for every $A,B\in\mathcal P(S)$.}
		\end{equation}
		Then $d_{\mathcal P(S)}$ is weighted by 
		$w\colon \mathcal P(S)\to \R_{\geq 0}$ defined by $w(A)=\lvert A\rvert$, for every $A\subseteq S$.  Moreover, a co-weight for $d_{\mathcal P(S)}$ is 
		$w'\colon \mathcal P(S)\to \R_{\geq0}$ defined by $w'(A)=|X|-|A|$ for every $A\in\mathcal P(S)$.
		\item Let $G$ be a finite abelian group and $L(G)$ be the family of all subgroups of $G$. We define on $L(G)$ the quasi-metric $d_{L(G)}$ as follows: 
		\begin{equation}\label{eq:d_L(G)}
		d_{L(G)}(H,K):=\log\lvert H+K:H\rvert, \quad \text{for every $H,K\in L(G)$.}
	\end{equation}
		Then $d_{L(G)}$ is weighted by $w\colon L(G)\to \R_{\geq0}$ defined by $w(H)=\log\lvert H\rvert$, for every $H\in L(G)$.  Moreover, a co-weight for $d_{L(G)}$ is  $w'\colon L(G)\to \R_{\geq0}$ defined by $w^\prime(H)=\log\lvert G\rvert-\log\lvert H\rvert$.
	\end{compactenum}
\end{example}

The quasi-metric spaces defined in the above example will be generalised in Example~\ref{ex:w_weighted_gqm}. 

\subsection{Weakly weighted generalised quasi-metric spaces}\label{ss:ww}

Let us focus on generalised quasi-metric spaces that are weakly weighted.

\begin{proposition}\label{prop:characterisation_ww_gqm}
	Let $(X,d)$ be a generalised quasi-metric space. Then $d$ is weakly weighted if and only if the following two properties hold:
	\begin{compactenum}[(a)]
		\item for every $x\in X$, $d\restriction_{\mathcal Q(x)}$ is weakly weighted; 
		\item for every $x,y\in X$, $d(x,y)=\infty$ if and only if $d(y,x)=\infty$.
	\end{compactenum}
\end{proposition}
\begin{proof}
	Let $\{X_i\}_{i\in I}$ be the family of connected components of $X$ and denote $d\restriction_{X_i}$ by $d_i$. If $d$ is weakly weighted by $w$, then every $d_i$ is weakly weighted (by $w\restriction_{X_i}$) and \eqref{eq:w_weight} implies property (b) as $w$ assumes only finite values.
	
	Conversely, suppose that (a) and (b) are fulfilled. Define, for every $x\in X$, 
	$w(x)=w_i(x)$, where $x\in X_i$ and $w_i$ is a weak weight for $d_i$. Then $w$ trivially satisfies \eqref{eq:w_weight}, for every 
	$x,y\in X$ with $x\cong_dy$. Moreover, condition (b) implies that $w$ satisfies the desired property also for points $x,y\in X$ with $x\not\cong_dy$.
\end{proof}

 To discuss  a corollary of Proposition \ref{prop:characterisation_ww_gqm}, let us describe the coproducts of the category ${\bf QMet}$ of generalised quasi-metric spaces and non-expansive maps between them.

\begin{remark}
Let $\{(X_i,d_i)\}_{i\in I}$ be a family of generalised quasi-metric spaces. Their  coproduct in {\bf QMet} is given by the disjoint union of the underlying 
	sets $X=\bigsqcup_{i\in I}X_i$ together with the generalised quasi-metric $d$ defined as follows: if, for every $i\in I$, $j_i\colon X_i\to X$ represents the canonical inclusion of $X_i$ into $X$, then
	$$d(j_i(x),j_k(y))=\begin{cases}
		\begin{aligned}
			& d_i(x,y) &\text{if $i=k$,}\\
			& \infty &\text{otherwise,} 
		\end{aligned}
	\end{cases}$$
	for every $j_i(x),j_k(y)\in X$.  
	Indeed, it can be proved that $(X,d)$ satisfies the universal property of coproducts.
	
	If each $d_i$ is weakly weighted by $w_i$, then $d$ is weakly weighted by the function $w$ making all the triangles in the following diagram  commute:
	$$\xymatrix{
		& & X\ar^w[d] & &\\
		& & \R & &\\
		\cdots & X_i\ar_{w_i}[ur]\ar^{j_i}[uur] & \cdots & X_k\ar^{w_k}[ul]\ar_{j_k}[uul] &\cdots
	}$$
\end{remark}

Thanks to the previous remark, Proposition \ref{prop:characterisation_ww_gqm} immediately implies the following characterisation.
\begin{corollary}
A generalised quasi-metric space is weakly weighted if and only if it is the coproduct of a family of  weakly weighted quasi-metric spaces.
\end{corollary}

 To conclude the subsection, let us focus on a specific example of generalised quasi-metric space. 
Given a directed graph $\Gamma=(V,E)$ and two vertices $x,y\in V$, a {\em directed path $P$ connecting $x$ to $y$} is a finite subset of edges $\{(z_{i-1},z_i)\}_{i=1,\dots,n}$ such that $z_0=x$ and $z_n=y$. We define the {\em path generalised quasi-metric} $d_\Gamma$ on $V$ as follows: for every $x,y\in V$,
$$
d_\Gamma(x,y):=\begin{cases}
	\begin{aligned}
		&\inf\{\lvert P\rvert\mid \text{$P$ is a directed path connecting $x$ to $y$}\} &\text{if $x\neq y$,}\\
		&0 &\text{otherwise.}
	\end{aligned}
\end{cases}$$
By definition, $d_\Gamma(x,y)=\infty$ in case there is no directed path from $x$ to $y$ in $\Gamma$. So, $d_\Gamma$ is a generalised quasi-metric as there may be no directed path connecting a vertex to another one. More precisely, $d_\Gamma$ is a quasi-metric if and only if $\Gamma$ is {\em strongly connected}, i.e., for every pair of vertices $x,y$ there is a path connecting $x$ to $y$ and a path connecting $y$ to $x$.

\begin{proposition}\label{prop:ww_dgraphs}
	Let $\Gamma=(V,E)$ be a strongly connected directed graph. Then the quasi-metric $d_\Gamma$ is weakly weighted if and only if $d_\Gamma$ is a metric, that is, $\Gamma$ is a non-directed graph (i.e., $(x,y)\in E$ if and only if $(y,x)\in E$).
\end{proposition}
\begin{proof}
	Clearly, if $d_\Gamma$ is a metric, then it is weakly weighted (a weak weight is any constant  function). Moreover, $d_\Gamma$ is a metric if and only if $\Gamma$ is non-directed. In fact, for every distinct $x,y\in V$, $(x,y)\in E$ if and only if $d_\Gamma(x,y)=1$ if and only if $d_\Gamma(y,x)=1$, which is equivalent to $(y,x)\in E$.
	
	Assume now that $d_\Gamma$ is weakly weighted by $w$. Define 
	$$C=\{d_\Gamma(x,y)+d_\Gamma(y,x)\mid (x,y)\in V\times V: d(x,y)> d(y,x)\}\subseteq\N$$
	(recall that $d_\Gamma$ has integer values). Suppose, by contradiction, that $C$ is non-empty. Then there exist $x,y\in V$ such that $d_\Gamma(x,y)>d_\Gamma(y,x)$ and $d_\Gamma(x,y)+d_\Gamma(y,x)$ is the minimum of $C$. Let $P$ and $Q$ be two directed paths going from $x$ to $y$ and from $y$ to $x$, respectively, of minimum length. Note that $d_\Gamma(x,y)\geq 2$, and so we can take a point $z\notin\{x,y\}$ that is crossed by the path $P$. Hence, $P$ is divided into two directed paths whose concatenation gives $P$: $P_1$ and $P_2$ going from $x$ to $z$ and from $z$ to $y$, respectively. Note that $\lvert P_1\rvert=d_\Gamma(x,z)$ and $\lvert P_2\rvert=d_\Gamma(z,y)$ as, otherwise, we could provide a directed path going from $x$ to $y$ that is strictly shorter than $P$. Moreover, $d_\Gamma(x,y)=d_\Gamma(x,z)+d_\Gamma(z,y)$. 
	By Lemma~\ref{lem:ww property}, one has $d_\Gamma(y,x)=d_\Gamma(y,z)+d_\Gamma(z,x)$. Then $d_\Gamma(x,z)+d_\Gamma(z,y)=d_\Gamma(x,y)>d_\Gamma(y,x)=d_\Gamma(y,z)+d_\Gamma(z,x)$ implies that either $d_\Gamma(x,z)>d_\Gamma(z,x)$ or $d_\Gamma(z,y)>d_\Gamma(y,z)$, and so a contradiction as $d_\Gamma(x,y)+d_\Gamma(y,x)$ was taken as minimal with that property.
\end{proof}

\begin{example}\label{ex:qm_not_weighted2}
	On $\TT$, we define the quasi-metric $d^\prime$ as the path quasi-metric associated to the following strongly connected directed graph:
	$$
	
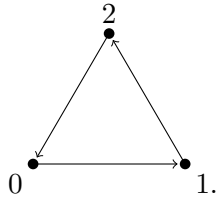
\begin{tikzpicture}
	\fill (0,0) circle (2pt) (2,0) circle (2pt) (1,1.73) circle (2pt);
	\draw[shorten >=0.1cm,->] (0,0)--(2,0);
	\draw[shorten >=0.1cm,->] (2,0)--(1,1.73); 
	\draw[shorten >=0.1cm,->] (1,1.73)--(0,0);
	\draw (0,0) node [below left]{$0$};
	\draw (2,0) node [below right]{$1.$};
	\draw (1,1.73) node [above]{$2$};  
	\end{tikzpicture}
	$$
	Then $d^\prime$ is not weakly weighted by Proposition~\ref{prop:ww_dgraphs}.
\end{example}

Thanks to Proposition~\ref{prop:characterisation_ww_gqm}, we can drop the request of strong connectivity from the statement of Proposition~\ref{prop:ww_dgraphs}.

\begin{corollary}\label{coro:ww_dgraphs}
	Let $\Gamma=(V,E)$ be a directed graph. Then $d_\Gamma$ is weakly weighted if and only if $d_\Gamma$ is a metric, that is, $\Gamma$ is a non-directed graph.
\end{corollary}



\subsection{Componentwisely weakly weighted generalised quasi-metric spaces}\label{ss:ww_gen}


Proposition~\ref{prop:characterisation_ww_gqm} and Corollary~\ref{coro:ww_dgraphs} suggest that the notion of weak weightedness is too restrictive in the realm of generalised quasi-metric spaces. So, inspired by Proposition~\ref{prop:characterisation_ww_gqm}, we propose a more suitable version.

\begin{definition}\label{def:componentwisely_ww_gen_qm}
	Let $(X,d)$ be a generalised quasi-metric space. Then $X$ is {\em componentwisely weakly weighted} if there exists a function $w\colon X\to\R$, called {\em componentwise weak weight} 
	for $(X,d)$ (or, 
	for $d$), such that \eqref{eq:w_weight} is fulfilled for every $x,y\in X$ satisfying $x\cong_dy$ (i.e., for every $x\in X$, $w\restriction_{\mathcal Q(x)}$ is a weak weight for $d\restriction_{\mathcal Q(x)}$).
\end{definition}

Clearly, a weakly weighted generalised quasi-metric space is, in particular, componentwisely weakly weighted, and the two notions coincide for quasi-metric spaces.

\begin{fact}\label{fact:cww_iff_each_Q_ww}
	A generalised quasi-metric space $(X,d)$ is componentwisely weakly weighted if and only if each connected component with the induced quasi-metric is weakly weighted (i.e., condition (a) of Proposition~\ref{prop:characterisation_ww_gqm} holds). More precisely, $w\colon X\to\R$ is a componentwise weak weight of $d$ if and only if, for every $x\in X$, $w\restriction_{\mathcal Q(x)}$ is a weak weight of the restriction of $d$ to ${\mathcal Q(x)}$.	
\end{fact}
\begin{proof}
	The ``only if'' implication is trivial. Conversely, the componentwise weak weight can be obtained by gluing together all the weak weights defined on the connected components. More explicitly, suppose that $\{(X_i,d_i)\}_{i\in I}$ is the partition of a generalised quasi-metric space $(X,d)$ in its connected components endowed with the inherited quasi-metrics that are weakly weighted by the corresponding functions $w_i\colon X_i\to\R$. Then the function $w\colon X\to\R$ defined by $w(x)=w_i(x)$ if $x\in X_i$ is a componentwise weak weight.
\end{proof}

According to Proposition~\ref{prop:ww_dgraphs}, a directed graph $\Gamma=(V,E)$ is componentwisely weakly weighted if and only if each strongly connected component (i.e., each connected component of $(V,d_\Gamma)$) is a non-directed graph. So, one can easily find an example of such a directed graph not satisfying condition (b) in Proposition~\ref{prop:characterisation_ww_gqm}, and so providing an example of a componentwisely weakly weighted generalised quasi-metric space that is not weakly weighted.

\medskip
In the following example we provide two generalised quasi-metric spaces that are componentwisely weakly weighted. However, for the proof of the property we wait for stronger results (see Example~\ref{ex:PX_and_LG_cww}).

\begin{example}\label{ex:w_weighted_gqm}
	\begin{compactenum}[(a)]
		\item For a set $S$,  we define a generalised quasi-metric $d_{\mathcal P(S)}$ on $\mathcal P(S)$ as in \eqref{eq:d_P(X)}.
		\item For an abelian group $G$, we define a generalised quasi-metric $d_{L(G)}$ on $L(G)$ as in \eqref{eq:d_L(G)}, with the abuse of notation $\log\infty=\infty$.  The large-scale geometry of the metric space $(L(G),d_{L(G)}^s)$ is studied in \cite{DikProZav}.
	\end{compactenum}
\end{example}

Even though the two examples provided in Example~\ref{ex:w_weighted_gqm} are componentwisely weakly weighted, they are not weakly weighted unless the set $S$ (respectively, the group $G$) is finite, i.e., unless they are quasi-metric spaces according to Proposition~\ref{prop:characterisation_ww_gqm}. In fact, if $S$ (respectively, $G$) is infinite, then the pair $\emptyset$ and $S$ (respectively, $\{0\}$ and $G$) does not satisfy condition (b) of the mentioned result.

\medskip
The equivalence relation $\sim$ from \eqref{sim} can be extended in the following way. Given a family $\{S_i\}_{i\in I}$ of pairwise disjoint subsets of a set $S$ (most commonly, a partition of $S$), two functions $f,g\colon S\to\R$ are equivalent relatively to the family $\{S_i\}_{i\in I}$, and we write 
\begin{equation}\label{simsim}
	f\approx_{\{S_i\}_i}g,\quad \text{if}\ f\restriction_{S_i}\sim g\restriction_{S_i}\ \text{for every $i\in I$.}
\end{equation}
If the family is clear from the context, we simply write $\approx$ for $\approx_{\{S_i\}_i}$.

\smallskip
Proposition~\ref{prop:w+c}(a) implies its extension to componentwise weak weights for generalised quasi-metric spaces.

\begin{corollary}\label{approx}
Let $X$ be a generalised quasi-metric space, and $w$ be a componentwise weak weight for $X$. Then the set of all componentwise weak weights for $X$ is $[w]_{\approx}$ where $\approx$ is relative to $\{\mathcal Q(x)\}_{x\in X}$.
\end{corollary}

\section{Generalisations of partial metrics}\label{pmsec}

\subsection{Partial metrics and their relation with weakly weighted quasi-metrics}\label{pmwwsec}

Partial metrics were introduced by Matthews~(\cite{Mat}). Then, O'Neill~(\cite{O96}) extended the notion maintaining the same terminology. In this paper we prefer to distinguish the two notions, and we call them partial metrics and weak partial metrics, respectively.

\begin{definition}[\cite{O96}]\label{def:w_partial_m}
	A {\em weak partial metric} $p$ on a non-empty set $X$ is a function $p\colon X\times X\to\R$ satisfying the following properties:
	\begin{compactenum}[(PM1)]
		\item for every $x,y\in X$, $x=y$ if and only if $p(x,x)=p(y,y)=p(x,y)$;
		\item for every $x\in X$, the function $p(x,\cdot)$ is minimised at $x$ (i.e., $p(x,x)\leq p(x,y)$, for every $y\in X$);
		\item $p$ is symmetric (i.e., for every $x,y\in X$, $p(x,y)=p(y,x)$);
		\item for every $x,y,z\in X$, $p(x,z)\leq p(x,y)+p(y,z)-p(y,y)$.
	\end{compactenum}
	A {\em partial metric}, according to~\cite{Mat}, is a weak partial metric $p$ on $X$ satisfying the further property $p\geq 0$. 
	
	The pair $(X,p)$ is a ({\em weak}) {\em partial metric space} if $p$ is a (weak) partial metric on $X$.
\end{definition}
If we consider the equivalence relation $\sim$ defined in \eqref{sim}, it is trivial to check that, if $(X,p)$ is a weak partial metric space, every function $p^\prime\colon X\times X\to\R$ satisfying $p^\prime\sim p$ is actually a weak partial metric. A similar result holds for partial metrics with the further assumption that $p^\prime\geq 0$.

Before introducing  some examples, let us show how weak partial metrics are tightly connected to weakly weighted quasi-metrics.
The corresponding connection between partial metrics and weighted quasi-metrics was proved in~\cite{Mat}, where the two notions were introduced.

\begin{theorem}\label{theo:wwqm_vs_wpm}
	Let $X$ be a non-empty set. 
	\begin{compactenum}[(a)]
		\item If $d$ is a quasi-metric on $X$ that is weakly weighted by $w$, then the function $p_{d,w}\colon X\times X\to\R$ defined by the law
		$$p_{d,w}(x,y):=d(x,y)+w(x),\quad\text{for every $x,y\in X$,}$$
		is a weak partial metric on $X$. If $w'$ is another weak weight for $d$, then $p_{d,w'}\sim p_{d,w}$.
		
		In particular, if $w$ is a weight (and so $d$ is weighted) then $p_{d,w}$ is a partial metric.
		\item If $p$ is a weak partial metric on $X$, then the function $d_p\colon X\times X\to\R_{\geq 0}$ defined by the law
		$$d_p(x,y):=p(x,y)-p(x,x),\quad\text{for every $x,y\in X$,}$$
		is a quasi-metric on $X$ which is weakly weighted by $w_p\colon X\to\R$ such that $w_p(x)=p(x,x)$ for every $x\in X$. Moreover, if $p^\prime\colon X\times X\to\R$ satisfies $p^\prime\sim p$, then 
		 $p^\prime$ is a weak partial metric, $d_{p'}=d_p$ and $w_{p'}\sim w_p$.
		
		In particular, if $p$ is a partial metric, then $w_p$ is a weight for $d_p$.
	\end{compactenum}
\end{theorem}
\begin{proof}
	(a) Let  $w$ be a weak weight for the quasi-metric $d$ on $X$. Properties (PM1) and (PM2) are easily verified for $p_{d,w}$ and (PM3) descends from \eqref{eq:w_weight}. Let us show (PM4). If $x,y,z\in X$, then the triangle inequality (QM2) yields
	$$\begin{aligned}p_{d,w}(x,z)&\,=d(x,z)+w(x)\leq d(x,y)+w(x)+d(y,z)+w(y)-w(y)=\\
		&\,=p_{d,w}(x,y)+p_{d,w}(y,z)-p_{d,w}(y,y).\end{aligned}$$
	
	The second assertion can be easily verified, in view of Proposition~\ref{prop:w+c}, whereas the last assertion is clear by definition.
	
	(b) Let $p$ be a weak partial metric. Property (PM2) implies that $d_p\geq 0$. Property (QM1) for $d_p$ is satisfied because of (PM1) and (PM3). Let us  prove property (QM2) by using (PM4). For every $x,y,z\in X$ one has
	$$d_p(x,z)=p(x,z)-p(x,x)\leq p(x,y)+p(y,z)-p(y,y)-p(x,x)=d_p(x,y)+d_p(y,z).$$
	Finally, by construction, $w_p$ is a weak weight for $d_p$.
	
	The second assertion can be easily verified, also using Proposition~\ref{prop:w+c},  whereas the last assertion is clear by definition.
\end{proof}


\begin{corollary}\label{coro:wwqm_vs_wpm_one_to_one}
	Let $X$ be a non-empty set. 
	If $d$ is quasi-metric on $X$ weakly weighted by $w$ and $p$ is a weak partial metric on $X$, then $$d_{p_{d,w}}=d\quad \text{and}\quad p_{d_p,w_p}=p.$$
	So, there exists a one-to-one correspondence between weakly weighted quasi-metrics (respectively, weighted quasi-metrics) on $X$ and equivalence classes of weak partial metrics (respectively, partial metrics) on $X$. 
\end{corollary}
\begin{proof} 
	Let us consider a quasi-metric $d$ weakly weighted by $w$. By Theorem~\ref{theo:wwqm_vs_wpm}(a), $[p_{d,w}]_\sim$ does not depend on the choice of $w$, so the map $d\mapsto [p_{d,w}]_\sim$ is well-defined. 
	Moreover, for every $x,y\in X$,
	$$d_{p_{d,w}}(x,y)=p_{d,w}(x,y)-p_{d,w}(x,x)=d(x,y)+w(x)-d(x,x)-w(x)=d(x,y).$$
	In particular, $d_{p_{d,w}}$ does not depend on the choice of $w$.
	
	Conversely, let us take two weak partial metrics $p\sim p'$ on $X$. According to Theorem~\ref{theo:wwqm_vs_wpm}(b), $d_p=d_{p'}$ and it is weakly weighted, so the map $[p]_\sim\mapsto d_p$ is well-defined. 
	Moreover, for every $x,y\in X$,
	\[p_{d_{p},w_{p}}(x,y)=d_p(x,y)+w_p(x)=p(x,y)-p(x,x)+p(x,x)=p(x,y).\qedhere\]
\end{proof}


\begin{example}\label{strings}
	Let $\Sigma^\ast$ be the family of all strings, both finite and infinite, on the alphabet $\Sigma$. For every pair of strings $s,s^\prime\in\Sigma$, we denote by $\lvert s\rvert$ the length of $s$ and by $l(s,s^\prime)$ the length of the longest common prefix of both $s$ and $s^\prime$. Then we define a partial metric $p$ on $\Sigma^\ast$ by the law $p(s,s^\prime)=2^{-l(s,s^\prime)}$, for every $s,s^\prime\in\Sigma^\ast$. 
	
	Using Theorem~\ref{theo:wwqm_vs_wpm} we can associate to $p$ the quasi-metric $d_p$ (that is, $d_p(s,s^\prime)=2^{-l(s,s^\prime)}-2^{-\lvert s\rvert}$, for every $s,s^\prime\in\Sigma^\ast$), which is weighted by the function $w_p$ (namely, $w_p(s)=2^{-\lvert s\rvert}$, for every $s\in\Sigma^\ast$).
\end{example}

The following example from \cite{AssKou} provides a weak partial metric space $(X,p)$ coming from a biological setting. The authors call the weak partial metric {\em strong} to emphasise that the request (PM2) in Definition~\ref{def:w_partial_m} is replaced by the following stronger one:
\begin{compactenum}
	\item[(PM2S)] for every $x,y\in X$, $p(x,x)<p(x,y)$.
\end{compactenum}

\begin{example}\label{ex:strong_pm}
	Let $\Sigma^{<\infty}$ be the family of all finite strings on the 
	(finite) alphabet $\Sigma$. For two strings $\overline x=(x_1,\dots,x_n),\overline y=(y_1,\dots,y_m)\in\Sigma^{<\infty}$, an {\em alignment of $\overline x$ and $\overline y$} is given by two strings $\overline x^\prime=(x_1^\prime,\dots,x_k^\prime)$ and $\overline y^\prime=(y_1^\prime,\dots,y_k^\prime)$ of the same length that can be obtained from $\overline x$ and $\overline y$, respectively, by adding occurrences of a new, blank, character $\#\notin\Sigma$. Given such an alignment $\overline x^\prime$ and $\overline y^\prime$, we compute its {\em score} $s^\prime(\overline x^\prime,\overline y^\prime)$ by comparing the two strings at each position $i\in\{1,\dots,k\}$ and adding the following real values:
	\begin{compactenum}[-]
		\item $\alpha$ if $x_i^\prime=y_i^\prime\neq\#$;
		\item $\beta$ if $x_i^\prime\neq y_i^\prime$ and both $x_i^\prime\neq\#$ and $y_i^\prime\neq \#$;
		\item $\gamma$ if precisely one of the two characters $x_i^\prime$ and $y_i^\prime$ is $\#$;
		\item $0$ if $x_i^\prime=y_i^\prime=\#$.
	\end{compactenum}

	For example, the alignment
$$\begin{aligned}\overline x^\prime&\,=G\,\#\,A\,\#\,T\,T\,A\,C\,A\,\#\quad \text{and}\\
	\overline y^\prime&\,=G\,C\,A\,\#\,T\,C\,A\,C\,G\,A 
\end{aligned}$$
of the strings $\overline x=GATTACA$ and $\overline y=GCATCACGA$ on the alphabet $\Sigma=\{G,A,T,C\}$ has score $s^\prime(\overline x^\prime,\overline y^\prime)=5\alpha+2\beta+2\gamma+0$.

	Then we assign to two strings $\overline x$ and $\overline y$ the following value:
	$$s(\overline x,\overline y)=\max\{s^\prime(\overline x^\prime,\overline y^\prime)\mid\text{$\overline x^\prime$ and $\overline y^\prime$ form an alignment of $\overline x$ and $\overline y$}\}.$$
	The {\em score scheme} just described is a useful technique to compare partial DNA strands (see~\cite{EidJonTay}).
	
	The function $p\colon\Sigma^{<\infty}\times\Sigma^{<\infty}\to\R$ defined by $p(\overline x,\overline y)=-s(\overline x,\overline y)$, for every $\overline x,\overline y\in\Sigma^{<\infty}$, is a strong weak partial metric provided that $\alpha>\beta$, $\alpha>\gamma$, $\beta\geq 2\gamma$ and $\gamma<0$ (see~\cite[Proposition 2.1]{AssKou}).
	
	If we note that, for every $\overline x\in\Sigma^{<\infty}$, $\overline x$ and $\overline y$ provide the best alignment of $\overline x$ and $\overline y$ themselves (the one with the highest score), we can characterise the weakly weighted quasi-metric $d_p$ induced by $p$ as follows: for every $\overline x,\overline y\in\Sigma^{<\infty}$, $d_p(\overline x,\overline y)=\alpha\lvert\overline x\rvert-s(\overline x,\overline y).$
\end{example}

\subsection{Generalised partial metrics and their relation with weakly weighted generalised quasi-metrics}\label{gpmwwsec}

Inspired by the notions of weak partial metric and generalised quasi-metric space, we introduce the following.

\begin{definition}
	A {\em generalised weak partial metric $p$} on a non-empty set $X$ is a function $p\colon X\times X\to\R\cup\{\infty\}$ satisfying the properties (PM1)--(PM4) of Definition~\ref{def:w_partial_m} (with the usual convention that $b<\infty+a=a+\infty=\infty$, for every $a,b\in\R$) and the following further one:
	\begin{compactenum}[(PM5)]
		\item for every $x\in X$, $p(x,x)<\infty$.
	\end{compactenum}
	The pair $(X,p)$ is called {\em generalised weak partial metric space}.
\end{definition}

Note that property (PM4) is well-defined thanks to (PM5).

Similarly to what we have done for generalised quasi-metric spaces, we can introduce an equivalence relation $\simeq_p$ on a generalised partial metric space $(X,p)$ as follows: for every $x,y\in X$, $x\simeq_py$ if $p(x,y)<\infty$. Then the equivalence classes of $\simeq_p$ are called {\em $p$-connected components}. Moreover, for every $x\in X$, $\mathcal Q^p_X(x)$ denotes the $p$-connected component to which $x$ belongs, i.e., $[x]_{\simeq_p}$. To distinguish between these two notions and the correspondent ones defined for generalised quasi-metric spaces, we emphasise the role of $p$ in the notation. 

As one may expect, also weakly weighted generalised quasi-metric spaces and generalised weak partial metric spaces are two faces of the same coin.

\begin{theorem}\label{theo:wwgqm_vs_gpm}
	Let $X$ be a non-empty set.
	\begin{compactenum}[(a)]
		\item If $d$ is a generalised quasi-metric on $X$ which is weakly weighted by $w$, then the function $p_{d,w}\colon X\times X\to\R\cup\{\infty\}$ defined by putting 
		$$p_{d,w}(x,y):=d(x,y)+w(x),\quad \text{for every $x,y\in X$,}$$
		is a generalised weak partial metric on $X$.
		Moreover, $\cong_d=\simeq_{p_{d,w}}$.
		\item If $p$ is a generalised weak partial metric on $X$, then the function $d_p\colon X\times X\to\R_{\geq 0}\cup\{\infty\}$ defined by
		$$d_p(x,y):=p(x,y)-p(x,x),\quad\text{for every $x,y\in X$,}$$
		is a generalised quasi-metric weakly weighted by the function $w_p\colon X\to\R$ defined by $w_p(x):=p(x,x)$, for every $x\in X$.
		Moreover, $\cong_{d_p}=\simeq_p$.
	\end{compactenum}
\end{theorem}
\begin{proof}
	The proof is an easy adaptation of that of Theorem~\ref{theo:wwqm_vs_wpm}. 
	To verify the coincidence of the equivalence relations, use the symmetry of the partial metrics.
\end{proof}

In order to have the complete extension of Theorem~\ref{theo:wwqm_vs_wpm}, we need to discuss the following.

Let $d$ be a generalised quasi-metric on a non-empty set $X$ weakly weighted by $w$ and $w^\prime$, and let $\{X_i\}_{i\in I}$ represent the partition of $X$ in its connected components with respect to $d$. They coincide with the $p$-connected components with respect to $p_{d,w}$ and $p_{d,w'}$ in view of Theorem~\ref{theo:wwgqm_vs_gpm}(a).
Then $p_{d,w}\approx_{\{X_i\times X_i\}_i}p_{d,w^\prime}$ (see \eqref{simsim} for the definition of $\approx$). 

On the other hand, let $p$ and $p^\prime$ be two generalised weak partial metrics on $X$, with $\simeq_p=\simeq_{p^\prime}$ and $p\approx_{\{X_i\times X_i\}_i}p^\prime$, where $\{X_i\}_{i\in I}$ is the family of all $p$-connected components of $X$. They coincide with the connected components with respect to $d_p$ and $d_{p'}$ by Theorem~\ref{theo:wwgqm_vs_gpm}(b). Then $d_p=d_{p^\prime}$ and $w_p\approx_{\{X_i\}_i}w_{p^\prime}$. 

Now, as in Corollary~\ref{coro:wwqm_vs_wpm_one_to_one}, we obtain  the following correspondence.

\begin{corollary}
	\label{cor:pm ww} 
Let $X$ be a non-empty set.
If $d$ is a generalised quasi-metric on $X$ weakly weighted by $w$ and $p$ is a generalised weak partial metric on $X$, then $d_{p_{w}}=d$ and $p_{d_p,w_p}= p$.
	
Hence, 
there is a one-to-one correspondence between weakly weighted generalised quasi-metrics on $X$ and equivalence classes with respect to  $\approx_{\{X_i\times X_i\}_i}$ 
of generalised weak partial metrics on $X$ inducing the same partition into ($p$-)connected components $\{X_i\}_{i\in I}$. 
\end{corollary}

Let us provide an example of a generalised weak partial metric (and so, according to Theorem~\ref{theo:wwgqm_vs_gpm}, of a weakly weighted generalised quasi-metric) inspired by Example~\ref{strings}.

\begin{example}\label{ex:gwpm}
	Let $\Sigma$ be an alphabet. Define the set
	$$\Sigma_{-\infty}^{\ast}=\{(x_n)_{n\in(-\infty,m)\cap\Z}\mid m\in\Z\cup\{\infty\},\,\forall n<m, x_n\in\Sigma\}.$$
	Let us fix some notation. For every $\overline x=(x_n)_{n\in(-\infty,m)\cap\Z}\in\Sigma_{-\infty}^{\ast}$, $L_{\overline x}\in\Z$ denotes the last index for which $\overline x$ is defined. More explicitly, $L_{\overline x}=m-1$. Then, for every $k\leq L_{\overline x}$, $\overline x[k]$ denotes the substring obtained by cutting $\overline x$ at the level $k$, i.e., $\overline x[k]=(x_n)_{n\in(-\infty,k]\cap\Z}$. Set now, for every $\overline x,\overline y\in\Sigma_{-\infty}^{\ast}$,
	$$l(\overline x,\overline y)=\sup\{k\in\Z\mid k\leq\min\{L_{\overline x},L_{\overline y}\},\,\overline x[k]=\overline y[k]\}.$$
	Then the function $p\colon\Sigma_{-\infty}^{\ast}\times\Sigma_{-\infty}^{\ast}\to\R\cup\{\infty\}$ defined by putting, for every $\overline x,\overline y\in\Sigma_{-\infty}^{\ast}$, $p(\overline x,\overline y)=2^{-l(\overline x,\overline y)}$ is a generalised weak partial metric on $\Sigma_{-\infty}^{\ast}$. If $\Sigma$ has at least two elements $\sigma_1,\sigma_2$, it is not a weak partial metric; in fact, the distance between the constant sequences $(\sigma_1)_{n\in\Z}$ and $(\sigma_2)_{n\in\Z}$ is infinite. Moreover, according to Theorem~\ref{theo:wwgqm_vs_gpm}, $d_p$ is a weakly weighted generalised quasi-metric.
\end{example}

\begin{remark}
	Let now $(X,d)$ be a generalised quasi-metric space and $w$ be a componentwise weak weight for it. Then it can be easily shown that the function $p_{d,w}$ defined as 
	$$p_{d,w}(x,y)=\begin{cases}
		\begin{aligned} & d(x,y)+w(x) &\text{if $x\cong_d y$,}\\
			& \infty &\text{otherwise,}
		\end{aligned}
	\end{cases}$$for every $x,y\in X$, is indeed a generalised weak partial metric. 
By Corollary~\ref{cor:pm ww}, it follows that every componentwise weak weight induces a weakly weighted generalised quasi-metric on $X$ which can be explicitly described as follows: for every $x,y\in X$,
	$$d_{p_{d,w}}(x,y)=\begin{cases}\begin{aligned}
			& d(x,y) &\text{if $x\cong_dy$,}\\
			& \infty &\text{otherwise.}
	\end{aligned}\end{cases}$$
	In particular, the identity map $id_X\colon(X,d_{p_{d,w}})\to(X,d)$ is non-expansive and $d_{p_{d,w}}$ is the coarsest weakly weighted generalised quasi-metric on $X$ having this property, i.e., if $d^\prime$ is a weakly weighted generalised quasi-metric on $X$ such that $id_X\colon(X,d^\prime)\to(X,d)$ is non-expansive, then $id_X\colon(X,d^\prime)\to(X,d_{p_{d,w}})$ is non-expansive. 
	
	For example, let $\Gamma=(V,E)$ be a directed graph endowed with the generalised quasi-metric $d_\Gamma$ and suppose that each strongly connected component of $\Gamma$ is a non-directed graph (see \S\ref{ss:ww} and \S\ref{ss:ww_gen}). In other words, $d_\Gamma$ is componentwisely weakly weighted, say by $w_\Gamma$.  In this case, the quasi-metric $d_{p_{d_{_\Gamma},w_{_\Gamma}}}$  coincides with $d_{\Gamma'}$, where $\Gamma'$ is obtained as disjoint union of all strongly connected components of $\Gamma$.
\end{remark}

\section{Generalised quasi-metric semilattices}
\label{sec:quasi-metric_semilattices}

Let $(X,\leq)$ be a {\em partially ordered set} (briefly, {\em poset}), i.e., a pair of a set $X$ and a {\em partial order} $\leq$ on it, which is a reflexive, transitive and antisymmetric relation. 
The {\em dual poset} $(X,\geq)$ consists of the same set $X$ and the \emph{inverse order} $\geq$ (i.e., $x\geq y$ if and only if $y\leq x$). A subset $Y$ of $X$ is said to be {\em convex} if, for every $x,y,z\in X$, $y\in Y$ provided that $x\leq y\leq z$ and $x,z\in Y$. 

A {\em meet-semilattice} (respectively, {\em join-semilattice}) is a poset $(X,\leq)$ such that for every $x,y\in X$ the {\em meet} $x\wedge y=\inf\{x,y\}$ (respectively, the {\em join} $x\vee y=\sup\{x,y\}$) exists. If a poset $(X,\leq)$ is simultaneously a meet- and a join-semilattice, then it is a {\em lattice}.

\begin{remark}\label{rem:duality_join_meet}
The poset $(X,\leq)$ is a meet-semilattice if and only if $(X,\geq)$ is a join-semilattice. More precisely, we have that $x\wedge_{\leq}y=x\vee_{\geq}y$, for every $x,y\in X$.
\end{remark}

We write simply {\em semilattice} to indicate either a meet-semilattice or a join-semilattice when the statement holds for both of them; 
in such a case we use an $\ast$ to indicate the semilattice operation (indeed a semilattice can be seen as a commutative semigroup where every object is idempotent). 
When  maintaining the distinction between meet- and join-semilattices is instead necessary, sometimes we write results explicitly only for meet-semilattices, while the dual statements can be similarly deduced using Remark~\ref{rem:duality_join_meet}. 

\smallskip
An equivalence relation $\cong$ on a semilattice $(X,\ast)$ is a {\em congruence} if, for every $x,y,z,w\in X$, $x\ast z\cong y\ast w$ provided that $x\cong y$ and $z\cong w$.

\begin{fact}\label{prop:congruence}
Let $(X,\ast)$ be a semilattice and $\cong$ a congruence on it. Then every equivalence class of $\cong$ is a convex subsemilattice of $X$.
\end{fact}
\begin{proof}
Let $x,y\in X$ satisfying $x\cong y$. Then $x\ast y\cong x\ast x=x$, thus $[x]_{\cong}$ is a subsemilattice of $X$. To prove that $[x]_{\cong}$ is convex, take two points $y,z\in X$ such that $x\leq y\leq z$ and $z\cong x$. Then $x\ast y\cong z\ast y$, and thus $x\cong y$ (if $\ast=\wedge$, then $x\wedge y=x$ and $z\wedge y=y$, while, if $\ast=\vee$, then $x\vee y=y$ and $z\vee y=z\cong x$).
\end{proof}


\subsection{The specialisation order of a generalised quasi-metric}\label{specordersec}


Every generalised quasi-metric space $(X,d)$ comes with the partial order $\leq_d$, called {\em specialisation order}, 
 i.e., 
 for $x,y\in X$, $x\leq_dy$ if and only if $d(x,y)=0$. 
\begin{remark}
 The partial order $\leq_d$ is indeed nothing but the specialisation order of the quasi-metric topology $\tau_d$ 
 induced 
 by $d$ on $X$.  Let us recall that 
$\tau_d$ is the $T_0$-topology generated by the base $\{B^d_\eps(x)\mid x\in X,\varepsilon>0\}$, where $B^d_\eps(x)=\{y\in X\mid d(x,y)<\eps\}$ is the $\varepsilon$-ball around the point $x\in X$.  Moreover, if $\tau$ is a topology on a set $X$, then the {\em specialisation order} $\leq_\tau$ of $\tau$ is defined as follows: for $x,y\in X$, $x\leq_\tau y$ if $y$ is contained in the closure of $\{x\}$.
\end{remark}

\begin{fact}[Monotonicity]\label{prop:mono_d}
Let $(X,d)$ be a generalised quasi-metric space and $x_1,x_2,y_1,y_2\in X$. Then 
\[d(x_1,y_1)\leq d(x_2,y_2)\ \text{provided}\ x_1\leq_d x_2\ \text{and}\ y_2\leq_d y_1.\]
\end{fact}
\begin{proof}
Take four points $x_1,x_2,y_1,y_2\in X$ such that $x_1\leq_d x_2$ and $y_2\leq_d y_1$. Then the triangle inequality (QM2) implies
\[d(x_1,y_1)\leq d(x_1,x_2)+d(x_2,y_2)+d(y_2,y_1)=d(x_2,y_2),\]
since $d(x_1,x_2)=0=d(y_2,y_1)$.
\end{proof}

\begin{remark}\label{rem:convex}
Given  a generalised quasi-metric space $(X,d)$, for every $x\in X$, $\mathcal Q(x)$ is convex.
In fact, fix $\overline x\in X$, take $x,y,z\in X$, and suppose that $x,z\in\mathcal Q(\overline x)$ and $x\leq_d y\leq_d z$. Since $z\in\mathcal Q(\overline x)=\mathcal Q(x)$ and  $d(y,z)=0$ by hypothesis, we get $d(y,x)\leq d(y,z)+d(z,x)=d(z,x)<\infty$. Hence, $y\in\mathcal Q(x)$ as $d(y,x)=0$ by assumption.
\end{remark}

We extend the following definition from~\cite{Sch96} (see also \cite{Sch_1} where it is called {\em order-convexity}) 
to every generalised quasi-metric space.

\begin{definition}\label{def:DPC}
For a generalised quasi-metric space $(X,d)$, the {\em descending path condition} is:
\begin{compactenum}
\item[(DPC)] for every 
$x,y,z\in X$ satisfying $x\geq_dy\geq_dz$, one has $d(x,z)=d(x,y)+d(y,z)$.
\end{compactenum}
\end{definition}

The following result immediately descends from the definition of {\rm(DPC)}.

\begin{fact}\label{DPCsub}
Let $(X,d)$ be a generalised quasi-metric space and $Y$ a subset of $X$. 
If $(X,d)$ satisfies {\rm(DPC)}, then $(Y,d\restriction_Y)$ satisfies {\rm(DPC)}.
\end{fact}

As  a straightforward 
consequence of Fact~\ref{DPCsub}, if a generalised quasi-metric space satisfies 
{\rm (DPC)}, 
then so does each connected component. Moreover, 
(DPC) 
can be separately verified in each connected component.

\begin{proposition}\label{DPCQ}
Let $(X,d)$ be a generalised quasi-metric space. Then $(X,d)$ satisfies {\rm(DPC)} if and only if each connected component $(\mathcal Q(x),d\restriction_{\mathcal Q(x)})$ 
of $X$ does it.
\end{proposition}
\begin{proof}
The `only if' implication follows from Fact~\ref{DPCsub}.  Conversely, suppose 
that $\mathcal Q(x)$ satisfies {\rm(DPC)} for every $x\in X$. 
Let $x\geq_dy\geq_dz$ be elements of $X$, for which (QM2) yields $d(x,z)\leq d(x,y)+d(y,z)$. If $d(x,z)=\infty$ then $d(x,z)=d(x,y)+d(y,z)$. Assume now that $d(x,z)<\infty$, and so $z\in \mathcal Q(x)$. Since $\mathcal Q(x)$ is convex by Remark~\ref{rem:convex}, also $y\in\mathcal Q(x)$. 
Thus $d(x,z)=d(x,y)+d(y,z)$ since $\mathcal Q(x)$ satisfies {\rm(DPC)}.
\end{proof}

Proposition~\ref{DPCQ} was proved in a more restrictive context in~\cite{DGBKTZ}.

\subsection{Invariant generalised quasi-metric semilattices}\label{invariantsec}

Given a poset $(X,\leq)$, we say that a (generalised) quasi-metric $d$ on $X$ is \emph{compatible} with the partial order of $X$ if $\leq_d=\leq$. 
In the rest of this section, we focus on generalised quasi-metrics defined on semilattices that are compatible with the algebraic structure of the  underlying space.

\begin{definition}\label{invdef}
A {\em generalised quasi-metric semilattice} is a generalised quasi-metric space $(X,d)$ such that the poset $(X,\leq_d)$ is a semilattice. 
A generalised quasi-metric meet-semilattice (respectively, join-semilattice) satisfying $d(x,y)=d(x,x\wedge y)$ (respectively, $d(x,y)=d(x\vee y,y)$) for all $x,y\in X$ is said to be {\em invariant}. 
\end{definition}

In the above definition of (generalised) quasi-metric semilattice, we do not require the semilattice operation to be quasi-uniformly continuous. This slight abuse of terminology is motivated by the fact that we mainly deal with invariant quasi-metric semilattices, which automatically meet the continuity requirement (see Proposition \ref{invth}(b)).

In the sequel, when we work with generalised quasi-metric semilattices $(X,\leq_d)$, we 
simply denote by $\leq$ the specialisation order $\leq_d$ if there is no risk of ambiguity.

\begin{remark}
As a consequence of Remark~\ref{rem:duality_join_meet}, $(X,d)$ is a generalised quasi-metric meet-semilattice if and only if $(X,d^{-1})$ is a generalised quasi-metric join-semilattice. More explicitly,  
$\leq_{d^{-1}}$ is the inverse order of $\leq_d$ and 
$x\wedge_{\leq_{d}}y=x\vee_{\leq_{d^{-1}}}y$, for every $x,y\in X$. Moreover, $(X,d)$ is invariant if and only if so is $(X,d^{-1})$.
\end{remark}

Similarly to {\rm(DPC)} that is inherited by subspaces, invariance is inherited by subsemilattices as stated in the following result.

\begin{fact}\label{fact:inv}
Let $X$ be a generalised quasi-metric semilattice and $Y$ a subsemilattice of $X$ endowed with the induced generalised quasi-metric. 
If $X$ is invariant, then so is $Y$.
\end{fact}

Let us see some useful characterisations of invariance. The proof can be obtained by adjusting that provided in \cite{DGBKTZ}.

\begin{proposition}\label{invth}
Let $(X,d)$ be a generalised quasi-metric meet-semilattice. Then the following properties are equivalent:
\begin{compactenum}[(a)]
\item $X$ is invariant;
\item for every $x,y,z\in X$, $d(z\wedge x,z\wedge y)\leq d(x,y)$ (i.e., for every $z\in X$, the shift $s_z\colon x\mapsto z\wedge x$ is non-expansive);
\item for every 
$x,y,z\in X$, $d(x,y\wedge z)\leq d(x,y)+d(x,z)$ (i.e., for every $x\in X$, $d(x,\cdot)\colon X\to\R_{\geq 0}\cup\{\infty\}$ is subadditive).
\end{compactenum}
\end{proposition}

For quasi-metric join-semilattices a similar result holds (as for item (c), it should be replaced with the following condition: {\em for every $x,y,z\in X$, $d(y\vee z,x)\leq d(y,x)+d(z,x)$}).

\smallskip
By invariance, we see that the connected components are subsemilattices:

\begin{proposition}\label{congd}
If $(X,d)$ is an invariant generalised quasi-metric semilattice, then $\cong_d$ is a congruence. In particular, every connected component $\mathcal Q(x)$ of $(X,d)$ is a convex invariant subsemilattice.
\end{proposition}
\begin{proof}
To prove that $\cong_d$ is a congruence, let $x,y,z,w\in X$ satisfying $x\cong_dy$ and $z\cong_dw$. Then, applying (QM2) and Proposition~\ref{invth}(b), we obtain that
$$d(x\ast z,y\ast w)\leq d(x\ast z,x\ast w)+d(x\ast w,y\ast w)\leq d(z,w)+d(x,y)<\infty,$$
and similarly $d(y\ast w,x\ast z)<\infty$. Finally, Facts~\ref{prop:congruence} and 
\ref{fact:inv} imply the second part of the statement.	
\end{proof}

The convexity of the connected components of a generalised quasi-metric semilattice was already verified 
in Remark~\ref{rem:convex}.

Let us mention that a (generalised) quasi-metric semilattice satisfying {\rm(DPC)} is not necessarily invariant (see \cite{DGBKTZ}  and Example~\ref{example}(a) below).

\medskip
Next we review  Examples~\ref{ex:w_weighted_qm},~\ref{ex:qm_not_weighted2},~\ref{strings} and~\ref{ex:strong_pm} to see whether the quasi-metric spaces 
introduced there are generalised quasi-metric semilattices and what properties they satisfy.

\begin{example}\label{example}
\begin{compactenum}[(a)]
\item Let us assume the notation used in Example~\ref{ex:w_weighted_qm}. The specialisation orders $\leq_{d_{\SSS}}$, $\leq_{d_{\TT}}$ and $\leq_{d_\R}$ coincide with the usual orders. Moreover, it is easy to see that $(\SSS,d_\SSS)$, $(\TT,d_\TT)$ and $(\R,d_\R)$ are both invariant quasi-metric meet-semilattices and invariant quasi-metric join-semilattices satisfying {\rm(DPC)}.

\item Let $(\mathbb T,d)$ be the quasi-metric space defined in Example~\ref{ex:qm_not_weighted1}. Then the specialisation order $\leq_d$ coincides with the usual order $0\leq 1\leq 2$. Moreover, $(\mathbb T,d)$ is both an invariant quasi-metric meet-semilattice and an invariant quasi-metric join-semilattice. However, it does not satisfy {\rm(DPC)}.

\item As for the family of all strings $\Sigma^\ast$ endowed with the quasi-metric $d_p$ (see Example~\ref{strings}), we can see that, for every $s,s^\prime\in\Sigma^\ast$, $s\leq_{d_p}s^\prime$ if and only if $s$ is a prefix of $s^\prime$ (i.e., $s$ coincides with an initial substring of $s^\prime$). Thus the meet of two strings $s$ and $s^\prime$ is given by the longest common prefix of $s$ and $s^\prime$. The pair $(\Sigma^\ast,d_p)$ is then an invariant quasi-metric meet-semilattice with {\rm(DPC)}.

\item Let  $(\Sigma^{<\infty},d_p)$ be the quasi-metric space defined in Example~\ref{ex:strong_pm}. Since the weak partial metric $p$ satisfies property (PM2S), $d_p$ satisfies the following strengthened version of (QM1): for every $x,y\in X$, $d_p(x,y)=0$ if and only if $x=y$. This property implies that the topology induced by $d_p$ satisfies the separation axiom T$_1$. Hence, in particular, 
the specialisation order $\leq_{d_p}$ coincides with the equality, and so $\Sigma^{<\infty}$ is not a semilattice.
\end{compactenum}
\end{example}

Let us now consider the generalised quasi-metric spaces defined in Example~\ref{ex:w_weighted_gqm} and already studied 
in~\cite{CDFGBT}. 

\begin{example}\label{ex:qm groups}
\begin{compactenum}[(a)]
\item For a 
set $X$, the specialisation order $\leq_{d_{\mathcal P(X)}}$ coincides with the containment 
$\supseteq$, which means, more explicitly, that, for every $A,B\subseteq X$, $d_{\mathcal P(X)}(A,B)=0$ (i.e., $A\leq_{d_{\mathcal P(X)}}B$) if and only $A\supseteq B$. With this partial order, $\mathcal P(X)$ is a lattice, where the meet of $A$ and $B$ is their union, while the join is their intersection. Moreover, $(\mathcal P(X),d_{\mathcal P(X)})$ is both an invariant generalised quasi-metric meet-semilattice with {\rm(DPC)} and an invariant generalised quasi-metric join-semilattice with {\rm(DPC)}. 

\item Given an abelian group $G$ and two subgroups $H,K\in L(G)$, $d_{L(G)}(H,K)=0$ (and so $H\leq_{d_{L(G)}}K$) if and only if $H\supseteq K$. Thus  $L(G)$ with this partial order  is 
a lattice, where, for every $H,K\in L(G)$, the meet of $H$ and $K$ is $H+K$ and the join is $H\cap K$. Then $(L(G),d_{L(G)})$ is both an invariant generalised quasi-metric meet-semilattice with {\rm(DPC)} and an invariant generalised quasi-metric join-semilattice with {\rm(DPC)}.
\end{compactenum}
\end{example}

\begin{remark}\label{rem:specialisation_order_by_p}
Let $X$ be a non-empty set. Any generalised weak partial metric $p$ on $X$ induces a partial order $\leq_p$ on $X$ defined as follows: if $x,y\in X$, then $x\leq_p y$ whenever $p(x,x)=p(x,y)$. It is trivial to see that $\leq_p=\leq_{p+c}$ for every constant $c\in\R$. The partial order $\leq_p$ is just the other side of the coin of the specialisation order. In fact, we have that $\leq_p=\leq_{d_p}$, and, moreover, if $d$ is a generalised quasi-metric on $X$ weakly weighted by $w$, then $\leq_d=\leq_{p_{d,w}}$ (see Theorem \ref{theo:wwgqm_vs_gpm} for the definitions of $d_p$ and $p_{d,w}$). In order to show it, let $x,y\in X$ and $w$ be a weak weight for $d$. Then the claim descends from the following two chains:
$$p(x,y)-p(x,x)=d_p(x,y),\text{ and } p_{d,w}(x,y)-p_{d,w}(x,x)=d(x,y)+w(x)-d(x,x)-w(x)=d(x,y).$$

Therefore, $(X,\leq_p)$ is a semilattice, and in this case we say that $(X,p)$ is a {\em generalised weak partial metric semilattice}, if and only if $(X,\leq_{d_p})$ is a semilattice. Furthermore, if $\ast$ is the semilattice operation, $p(x,y)=p(x,x\ast y)$ for every $x,y\in X$ if and only if $d_p$ is invariant. For consistency, we refer to this property by saying that $p$ is {\em invariant}. Similarly, a generalised quasi-metric semilattice $(X,d)$ weakly weighted by $w$ is invariant if and only if so is $p_{d,w}$. Thus, Corollary \ref{cor:pm ww} leads to a similar correspondence between invariant weakly weighted (generalised) quasi-metric semilattices and invariant (generalised) weak partial metric semilattices.
\end{remark}

\section{Weakly weighted quasi-metric semilattices and their characterisations}\label{wwqmsec}

The notion of weak weightedness for generalised quasi-metric semilattices (and so that of generalised weak partial metric semilattice according to Remark \ref{rem:specialisation_order_by_p}) is too restrictive, as the following remark shows. Motivated by that observation, in this section we focus on weakly weighted quasi-metric semilattices.

\begin{remark}\label{rem:gqm_semil_not_weighted}
Let $(X,d)$ be a generalised quasi-metric meet-semilattice (the other case can be similarly treated). Suppose that there exist $x,y\in X$ such that $d(x,y)=\infty$ (and so $d$ is not a quasi-metric).
Fact~\ref{prop:mono_d} implies that $d(x,y)\leq d(x,x\wedge y)$, and so $d(x,x\wedge y)=\infty$. Moreover, $d(x\wedge y,x)=0$, and so $(X,d)$ cannot be weakly weighted in view of Proposition~\ref{prop:characterisation_ww_gqm}. 
\end{remark}

\subsection{Weakly weighted quasi-metric semilattices and (DPC)}\label{ss:weighting}

\begin{proposition}\label{prop:weak_wei_oc}
A quasi-metric space $(X,d)$ satisfies {\rm (DPC)} provided that it is weakly weighted.
\end{proposition}
\begin{proof}
Let $w$ be a weak weight for $d$, and $x,y,z\in X$ such that $x\geq_d y\geq_d z$. By Lemma~\ref{lem:ww property}, one has that $d(x,z)=d(x,y)+d(y,z)$,
which implies that {\rm(DPC)} holds.
\end{proof}

We show in Theorem~\ref{thm:weak_weight} that the statement of Proposition~\ref{prop:weak_wei_oc} can be reversed 
for invariant quasi-metric semilattices. To this end, we use the following functions and the subsequent lemma describing them.

Let $(X,d)$ be an invariant quasi-metric semilattice satisfying {\rm(DPC)}, and $x\in X$. We define the function $w_x\colon X\to\R$ as follows: 
\begin{equation}\label{eq:w_x}w_x(y):=d(x,y)-d(y,x), \quad \text{for every $y\in X$.}
\end{equation}

\begin{lemma}\label{lemma:DPC_to_weight}
Let $(X,d)$ be an invariant quasi-metric meet-semilattice satisfying {\rm(DPC)}. Then, 
$$w_x(y)=d(x,x\wedge y\wedge z)-d(y,x\wedge y\wedge z), \quad \text{for every $x,y,z\in X$.}$$
\end{lemma}
\begin{proof}
	Let $x,y,z\in X$. Then invariance and {\rm(DPC)} imply that
	\begin{align*}
	w_x(y)&\,=d(x,y)-d(y,x)=d(x,x\wedge y)-d(y,x\wedge y)=\\
	&\,=d(x,x\wedge y)+d(x\wedge y,x\wedge y\wedge z)-d(y,x\wedge y)-d(x\wedge y,x\wedge y\wedge z)=\\
	&\,=d(x,x\wedge y\wedge z)-d(y,x\wedge y\wedge z).\qedhere
	\end{align*}
\end{proof}

\begin{theorem}\label{thm:weak_weight}
	Let $(X,d)$ be an invariant quasi-metric semilattice satisfying {\rm(DPC)}. Then, for every $x\in X$, $w_x$ is a weak weight for $d$ (so, $(X,d)$ is weakly weighted).
	 In particular, for every $x,y\in X$, $w_x\sim w_y$.
\end{theorem}
\begin{proof}
 We prove the statement assuming that $(X,d)$ is a meet-semilattice.
	Fixed $x\in X$, we claim that, for every $y,z\in X$, $d(y,z)+w_x(y)=d(z,y)+w_x(z)$. 
	Applying Lemma~\ref{lemma:DPC_to_weight} and {\rm(DPC)}, we obtain that 
	\begin{align*}
	w_x(y)-w_x(z)&\,=d(x,x\wedge y\wedge z)-d(y,x\wedge y\wedge z)-d(x,x\wedge y\wedge z)+d(z,x\wedge y\wedge z)=\\
	&\,=d(z,x\wedge y\wedge z)-d(y,x\wedge y\wedge z)=w_z(y)=d(z,y)-d(y,z),
	\end{align*}
	which concludes the verification that $w_x$ is a weak weight for $d$.
 The last assertion follows from Proposition~\ref{prop:w+c}(a).
\end{proof}

So, Proposition~\ref{prop:weak_wei_oc} and Theorem~\ref{thm:weak_weight} immediately give  the following result.

\begin{corollary}\label{theo:oc_vs_weak_wei}
An invariant quasi-metric semilattice is weakly weighted if and only if it satisfies {\rm(DPC)}. 
\end{corollary}

Let us consider the examples contained in Example~\ref{ex:w_weighted_qm} satisfying {\rm(DPC)}.

\begin{example}\label{ex:q_m_semilattices_wei}
\begin{compactenum}[(a)]
\item We start from the Sierpi\'nski space $\mathbb S=\{0,1\}$. Since $(\mathbb S,d_{\mathbb S})$ is an invariant quasi-metric semilattice with {\rm(DPC)}, according to Theorem~\ref{thm:weak_weight}, it is weakly weighted. Moreover, 
$$w_1(x)=d_\mathbb S(1,x)-d_\mathbb S(x,1)=d_\mathbb S(1,x)= 1-x,$$
which is the weight provided in Example~\ref{ex:w_weighted_qm}(a).

\item Analogously to (a),  we recover the weight $w$ defined in Example~\ref{ex:w_weighted_qm}(b) for $(\T,d_\T)$ as $w_2$.

\item In Example~\ref{ex:w_weighted_qm}(c) we have provided a weak weight $w$ for the quasi-metric space  $(\R,d_\R)$ with $d_\R(x,y)=\max\{x-y,0\}$ for every $x,y\in\R$. A straightforward computation gives $w=w_0$. 
\end{compactenum}
\end{example}

Next we review the examples contained in Example~\ref{entex}.

\begin{example} 
\begin{compactenum}[(a)]
\item As in Example~\ref{entex}(a), let $X$ be a finite set and on the power set $\mathcal P(X)$ consider the quasi-metric $d_{\mathcal P(X)}$. 
We know that $d_{\mathcal P(X)}$ is weighted by the weight $w(A)=\lvert A\rvert$ for every $A\subseteq X$, 
and now we recover it as 
$w=w_\emptyset$. 

\item Let $G$ be a finite abelian group and $d_{L(G)}$ be the quasi-metric on $L(G)$ defined in Example~\ref{entex}(b). 
We have already noticed in the mentioned example that $d_{L(G)}$ is weighted by the weight $w(H)=\log\lvert H\rvert$ for every subgroup $H$ of $G$, and now we have that 
$w=w_{\{0\}}$.
\end{compactenum}
\end{example}

Now we present some important applications of Theorem~\ref{thm:weak_weight}.

\begin{theorem}\label{thm:weight}
	Let $(X,d)$ be an invariant quasi-metric semilattice satisfying {\rm(DPC)}. 
	Then the following properties are equivalent:
	\begin{compactenum}[(a)]
		\item $(X,d)$ is weighted;
		\item for every $x\in X$, there exists $c_x\in\R$ such that  $
		w_x(y)\geq c_x$, for every $y\in X$;
		\item there exist $x\in X$ and $c\in\R$ such that  $
		w_x(y)\geq c$, for every $y\in X$.
	\end{compactenum}
If $(X,d)$ is a join-semilattice, then also the following condition is equivalent:
\begin{compactenum}[(d)]
	\item $(X,d)$ is an {\em $M$-space} (\cite[Definition~2.1]{Sch_1}), i.e., for every $x\in X$, there exists  $c_x^\prime\in\R_{\geq0}$ such that, for every $y\in X$ satisfying $x\leq y$, $d(y,x)\leq c_x^\prime$.
\end{compactenum}
\end{theorem}
\begin{proof}
Theorem~\ref{thm:weak_weight} implies that, for every $x\in X$, $w_x$ is a weak weight for $d$. Then the equivalences (a)$\Leftrightarrow$(b)$\Leftrightarrow$(c) follow from Proposition~\ref{prop:bound}(a).
	
(b)$\Rightarrow$(d)  Let $x\in X$ and $c_x$ as in (b). If $y\in X$ is such that $x\leq y$, then $c_x\leq w_x(y)=d(x,y)-d(y,x)=-d(y,x)$; in particular $c_x\leq 0$.  Hence, $d(y,x)\leq -c_x=:c_x^\prime\geq0$. 

(d)$\Rightarrow$(b) Let $x,y\in X$, and $c_x^\prime\in\R$ satisfying the hypothesis in (d). Then invariance implies that \[w_x(y)=d(x,y)-d(y,x)\geq -d(y,x)=-d(y\vee x,x)\geq -c_x^\prime=:c_x.\qedhere\]
\end{proof}

The equivalence of (a) and (d) in the above theorem provides another proof to \cite[Theorem 2.31]{Sch_1}.

We can analogously obtain the counterpart of Theorem~\ref{thm:weight} characterising co-weighted quasi-metric semilattices.
 The equivalence of (a) and (d) in the theorem below covers~\cite[Theorem 2.34]{Sch_1}.

\begin{theorem}\label{coro:co_weight_mspace}
	Let $(X,d)$ be an invariant quasi-metric semilattice satisfying {\rm(DPC)}. Then the following properties are equivalent:
	\begin{compactenum}[(a)]
		\item $(X,d)$ is co-weighted;
		\item for every $x\in X$, there exists $c_x\in\R$ such that  $
		w_x(y)\leq c_x$, for every $y\in X$;
		\item there exist $x\in X$ and $c\in\R$ such that  $
		w_x(y)\leq c$, for every $y\in X$.
	\end{compactenum}
If $(X,d)$ is a meet-semilattice, then also the following condition is equivalent:
\begin{compactenum}[(d)]
	\item $X$ is an {\em $m$-space} (\cite[Definition~2.10]{Sch_1}), i.e., for every $x\in X$, there exists $c_x^\prime\in\R_{\geq0}$ such that, for every $y\in X$ satisfying $y\leq x$, $d(x,y)\leq c_x^\prime$.
\end{compactenum}
\end{theorem}
\begin{proof}
	Theorem~\ref{thm:weak_weight} implies that, for every $x\in X$, $w_x$ is a weak weight for $d$. Then the equivalences (a)$\Leftrightarrow$(b)$\Leftrightarrow$(c) follow from Proposition~\ref{prop:bound}(b). 
	
	(b)$\Rightarrow$(d) Let $x\in X$ and $c_x$ as in (b). Then, if $y\in X$ satisfies $y\leq x$, 
	$$c_x\geq w_x(y)=d(x,y)-d(y,x)=d(x,y).$$
	
	(d)$\Rightarrow$(b) Let $x,y\in X$, and $c_x^\prime\in\R$ satisfying the hypothesis in (d). Then, since $X$ is invariant, and $x\wedge y\leq x$,
	\[w_x(y)=d(x,y)-d(y,x)\leq d(x,y)=d(x,x\wedge y)\leq c_x^\prime.\qedhere\]
\end{proof}

\begin{remark}\label{topbottom}
Let us consider two particular cases of 
Theorems~\ref{thm:weight} and~\ref{coro:co_weight_mspace}.
Let $(X,d)$ be an invariant quasi-metric semilattice satisfying {\rm(DPC)} and with a bottom element $\perp$. The definition of the function $w_\perp$ can be simplified as follows: for every $x\in X$, since $\perp\leq x$, $w_\perp(x)=d(\perp,x)-d(x,\perp)=-d(x,\perp)\leq 0$. Hence, 
$(X,d)$ is co-weighted. If, otherwise $X$ has a top element $\top$, for every $x\in X$, $w_\top(x)=d(\top,x)\geq 0$, 
and so $(X,d)$ is weighted.

Clearly, the functions $w_\perp$ and $w_\top$ can be considered also in semilattices that do not satisfy {\rm(DPC)}. Then the following results can be proved.
\begin{compactenum}[(a)]
\item Let $X$ be a quasi-metric meet-semilattice.
\begin{compactenum}
	\item[(a$_1$)] If $X$ has a bottom element $\perp$, then $w_\perp$ 
	is {\em supadditive} (i.e., $w_\perp(x\wedge y)\geq w_\perp(x)+w_\perp(y)$ for every $x,y\in X$) according to Fact~\ref{prop:mono_d}.
	\item[(a$_2$)] If $X$ has a top element $\top$ and $X$ is invariant, then $w_\top$ 
	is subadditive because of Proposition~\ref{invth}.
\end{compactenum}
\item Let $X$ be a quasi-metric join-semilattice.
\begin{compactenum}
	\item[(b$_1$)] If $X$ has a top element $\top$, then $w_\top$ is subadditive thanks to Fact~\ref{prop:mono_d}. 
	\item[(b$_2$)] If $X$ has a bottom element $\perp$ and $X$ is invariant, then $w_\perp$ is supadditive according to Proposition~\ref{invth}.
\end{compactenum}
\end{compactenum}	
\end{remark}

\subsection{Semivaluations, semi-co-valuations and weak weights}\label{ss:semivaluation}

We start recalling the definition of semivaluation from~\cite{Nak70}, giving also the corresponding notion of semi-co-valuation. These were proposed with the following terminology in~\cite{Sch} with domain $\R_{\geq0}$ instead of the whole $\R$ and using equivalent defining conditions.

\begin{definition}\label{def:val_and_coval}
Let $(X,\leq)$ be a meet-semilattice. A function $f\colon X\to\R$ is called:
\begin{compactenum}[-]
\item {\em meet valuation} if $f(x)+f(x\wedge y\wedge z)\geq f(x\wedge y)+f(x\wedge z)$, for every $x,y,z\in X$;
\item {\em meet co-valuation} if $f(x)+f(x\wedge y\wedge z)\leq f(x\wedge y)+f(x\wedge z)$, for every $x,y,z\in X$.
\end{compactenum}
Let $(X,\leq)$ be a join-semilattice. A function $f\colon X\to\R$ is called:
\begin{compactenum}[-]
	\item {\em join valuation} if $f(x)+f(x\vee y\vee z)\leq f(x\vee y)+f(x\vee z)$, for every $x,y,z\in X$;
	\item {\em join co-valuation} if $f(x)+f(x\vee y\vee z)\geq f(x\vee y)+f(x\vee z)$, for every $x,y,z\in X$.
\end{compactenum}
\end{definition}

We refer to \cite{DGBKTZ} for other characterisations of join valuations, and the respective equivalent conditions for the other properties can be easily and similarly obtained.

The following definition, as given in \cite{Nak70} and \cite{Sch}, conveniently regroups the notions provided in Definition~\ref{def:val_and_coval}. 

\begin{definition}
A function  $f\colon X\to \R$ on a semilattice $(X,\leq)$  is called a {\em semivaluation} (respectively, {\em semi-co-valuation}) if it is either a meet or a join valuation (respectively, either a meet or a join co-valuation). 
\end{definition}
In other words, the prefixes join and meet before the term (co-)valuation appear when we want to emphasise the type of semilattice  with which we are working. In the remaining cases we use the more general term semi(-co-)valuation.

\begin{proposition}\label{prop:mono_semival+}
Let $(X,\leq)$ be a  semilattice and $f\colon X\to \R$ a function.
\begin{compactenum}[(a)]
\item The function $f$ is a semivaluation (respectively, semi-co-valuation) if and only if $-f$ is a semi-co-valuation (respectively, semivaluation).
\item If $f$ is a semi(-co-)valuation, then every $f'\in[f]_\sim$ has the same property. 
\item If $f$ is a semivaluation, then $f$ is non-decreasing.
\item If $f$ is a semi-co-valuation, then $f$ is non-increasing.
\end{compactenum}
\end{proposition}
\begin{proof}
Items (a) and (b) are trivial, while item (c) descends from items (a) and (d). 

We want to prove item (d). Assume that $X$ is a meet-semilattice and $f$ a meet co-valuation. If $x,y\in X$ satisfy $x\leq y$, then
$$f(y)+f(x)=f(y)+f(y\wedge x\wedge x)\leq f(y\wedge x)+f(y\wedge x)=f(x)+f(x),$$
and so $f(y)\leq f(x)$, that is, $f$ is non-increasing. If $X$ is a join-semilattice and $f$ is a join co-valuation, then the claim can be similarly shown.
\end{proof}

As for the monotonicity results provided in Proposition~\ref{prop:mono_semival+}(c) and (d), proofs can be found in \cite[Lemma 4]{Sch} and \cite[Lemma 3.1(1)]{Nak70}. 
The same items of Proposition~\ref{prop:mono_semival+} also imply that 
a semi(-co-)valuation is strictly monotone if and only if it is injective.

\smallskip
Let us now start describing the connection between weak weights and and semi(-co-)valuations.

\begin{theorem}\label{theo:w_wei_to_valu_join_1}
Let $(X,d)$ be an invariant quasi-metric semilattice and let $w$ be a weak weight for $d$. Then $w$ is a strictly decreasing semi-co-valuation, while $-w$ is a strictly increasing semivaluation.
\end{theorem}
\begin{proof}
We prove the result for $(X,d)$ a meet-semilattice.
Let $x,y,z\in X$. Invariance and Proposition~\ref{invth}(b) imply that $d(x\wedge z,x\wedge y\wedge z)\leq d(x,x\wedge y)$. Since $d(x\wedge y\wedge z,x\wedge z)=0=d(x\wedge y,x)$, \eqref{eq:w_weight} implies that
$$w(x\wedge y\wedge z)-w(x\wedge z)=d(x\wedge z,x\wedge y\wedge z)\leq d(x,x\wedge y)=w(x\wedge y)-w(x),$$
and so $w$ is a meet co-valuation. Moreover, $w$ is strictly decreasing as, if $x< y$, then $d(y,x)>0$ thanks to (QM1) since $d(x,y)=0$, and so 
\[w(y)< w(y)+d(y,x)=d(x,y)+w(x)=w(x).\]
That $-w$ is a strictly increasing meet valuation follows from Proposition~\ref{prop:mono_semival+}(a) and (c).
\end{proof}

When a semilattice admits a semi(-co-)valuation, it can be equipped with a distance function as follows.

\begin{definition}\label{def:d_f}
If $X$ is a meet-semilattice and $f$ is a meet co-valuation, we define $d_f\colon X\times X\to\R_{\geq 0}$ as follows: $$d_f(x,y):=f(x\wedge y)-f(x),\quad \text{for every $x,y\in X$.}$$

Dually, if $X$ is a join-semilattice and $f$ is a join co-valuation, we define $d_f\colon X\times X\to\R_{\geq 0}$ as follows: $$d_f(x,y):=f(y)-f(x\vee y), \quad \text{for every $x,y\in X$.}$$
\end{definition}

Definition~\ref{def:d_f} generalises that given in \cite[Theorem 11]{Sch}.

\smallskip
The following result shows how to produce weakly weighted invariant quasi-metrics from semi(-co-)valuations, 

\begin{theorem}\label{theo:w_wei_to_valu_2}
Let $(X,\leq)$ be a semilattice 
and $f\colon X\to \R$. If $f$ is a strictly decreasing semi-co-valuation, then $d_f$ is an invariant quasi-metric on $X$ satisfying $\leq=\leq_{d_f}$, and weakly weighted by $f$.

If $f$ is a strictly increasing semivaluation, then $d_{-f}$ is an invariant quasi-metric on $X$ satisfying $\leq=\leq_{d_{-f}}$, and weakly weighted by $-f$.
\end{theorem}
\begin{proof}
Assume that $(X,\leq)$ is a meet-semilattice and so $f$ is a meet co-valuation. Let us first prove that $d_f$ is a quasi-metric.
Since $f$ is decreasing, $d_f$ assumes non-negative values. To prove (QM1), take two points $x,y\in X$ such that $d_f(x,y)=d_f(y,x)=0$. Then the definition of $d_f$ implies that $f(x)=f(x\wedge y)=f(y)$. As $f$ is strictly monotone, the equalities $x=x\wedge y=y$ descend. To verify the triangular inequality (QM2), let $x,y,z\in X$. 
Due to the fact that $f$ is a non-increasing 
meet co-valuation,
$$f(y)+f(x\wedge z)\leq f(y)+f(x\wedge z\wedge y)\leq f(x\wedge y)+f(z\wedge y),$$
and so,  again since $f$ is non-increasing,
\begin{align*}d_f(x,z)&\,=f(x\wedge z)-f(x)=f(x\wedge y)-f(x)+f(x\wedge z)-f(x\wedge y)=\\
&\,=d_f(x,y)+f(x\wedge z)-f(x\wedge y)\leq d_f(x,y)+f(x\wedge z)-f(y)\leq d_f(x,y)+d_f(y,z).\end{align*}

Note that $x\leq_{d_f}y$ if and only if $f(x)=f(x\wedge y)$, which is then equivalent to $x=x\wedge y$ and therefore to $x\leq y$. Hence, $\leq=\leq_{d_f}$. 
Finally, $d_f$ is invariant by construction and it can be easily seen that $f$ is a weak weight for $d_f$.

The second part of the theorem follows from the first one. In fact, if $f$ is a strictly increasing semivaluation, then $-f$ is a strictly decreasing semi-co-valuation, by Proposition~\ref{prop:mono_semival+}(a). Thus, we can associate to $f$ the map $d_{-f}$. 
\end{proof}

\begin{corollary}\label{coro:w_wei_to_valu2}
Let $(X,\leq)$ be a semilattice. There are one-to-one correspondences between the three sets consisting of the following objects:
	\begin{compactenum}[(a)]
		\item invariant weakly weighted quasi-metrics $d$ on $X$ satisfying $\leq_d=\leq$;
		\item equivalence classes of strictly decreasing semi-co-valuations $f$ 
		on $X$;
		\item equivalence classes of strictly increasing semivaluations $g$ 
		on $X$.
	\end{compactenum}
	More precisely, to $d$ we associate $[w]_\sim$ and $[-w]_\sim$, where $w$ is a weak weight for $d$, to $[f]_\sim$ we associate $d_f$ and $[-f]_\sim$, and, finally, to $[g]_\sim$ we associate $d_{-g}$ and $[-g]_\sim$.
\end{corollary}
\begin{proof} 
Let $(X,\leq)$ be a meet-semilattice. We prove 	the correspondence between (a) and (b) as that between (b) and (c) is trivial by Theorem~\ref{theo:w_wei_to_valu_2} and Proposition~\ref{prop:mono_semival+}(a).
Both associations described in the statement are well-defined, respectively, one by virtue of Theorem~\ref{theo:w_wei_to_valu_join_1} and Proposition~\ref{prop:w+c}(a), and the converse one by Theorem~\ref{theo:w_wei_to_valu_2} since two semi-co-valuations $f,g$ of $X$ satisfy $d_f=d_g$ provided that $f\sim g$. 

Moreover, we have to verify that the two associating maps are one the inverse of the other. Indeed, given a strictly decreasing semi-co-valuation $f$ of $X$, we have seen in Theorem~\ref{theo:w_wei_to_valu_2} that $f$ is a weak weight for $d_f$ and to $d_f$ we associate $[f]_\sim$. 
On the other hand, we check that if $d$ is an invariant quasi-metric on $X$ with $\leq_d=\leq$ and weakly weighted by $w$, then $d=d_w$. For $x,y\in X$, since $d(x\wedge y,x)=0$ and $d$ is invariant,
\[d_w(x,y)=w(x\wedge y)-w(x)=w(x\wedge y)+d(x\wedge y,x)-w(x)=d(x,x\wedge y)=d(x,y).\qedhere\] 
\end{proof}

As particular cases of Corollary~\ref{coro:w_wei_to_valu2}, we find the following classical results due to Schellekens~\cite{Sch} since each (co-)weighted quasi-metric admits a unique fading (co-)weight in view of Remark~\ref{fading}.

\begin{corollary}
	\begin{compactenum}[(a)]
		\item {\rm(\cite[Theorem 11]{Sch})} Let $(X,\leq)$ be a meet-semilattice. There are one-to-one correspondences between the three sets consisting of the following objects:
		\begin{compactenum}[(i)]
		\item invariant co-weighted quasi-metrics $d$ on $X$ with $\leq_d=\leq$;
		\item fading strictly increasing meet valuations 
		on $X$;
		\item fading strictly decreasing meet co-valuations 
		on $X$.
		\end{compactenum}
		\item {\rm(\cite[Theorem 10]{Sch})} Let $(X,\leq)$ be a join-semilattice. There are one-to-one correspondences between the three sets consisting of the following objects:
		\begin{compactenum}[(i)] 
		\item invariant weighted quasi-metrics $d$ on $X$ with $\leq_d=\leq$;
		\item fading strictly decreasing join co-valuations 
		on $X$;
		\item fading strictly increasing join valuations 
		on $X$. 
		\end{compactenum}
	\end{compactenum}
\end{corollary}

\section{Componentwisely weakly weighted generalised quasi-metric semilattices and their characterisations}\label{ss:wwgen}

Remark~\ref{rem:gqm_semil_not_weighted} shows that the notion of weak weightedness does not suit generalised quasi-metric semilattices. However, in what follows we prove characterisations of componentwisely weakly weighted generalised quasi-metric semilattices similar to those obtained for weakly weighted quasi-metric semilattices.


Proposition~\ref{prop:weak_wei_oc} extends to generalised quasi-metric spaces as follows. 

\begin{proposition}\label{prop:weak_wei_gen_to_oc}
If $(X,d)$ is a componentwisely weakly weighted generalised quasi-metric space, then $(X,\leq_d)$ satisfies {\rm(DPC)}.
\end{proposition}
\begin{proof}
It immediately follows from Fact~\ref{fact:cww_iff_each_Q_ww}, and Propositions~\ref{prop:weak_wei_oc} and~\ref{DPCQ}. 
\end{proof}

Similarly to the previous result, the following theorem can be proved applying Propositions~\ref{prop:weak_wei_gen_to_oc} and~\ref{DPCQ}, Theorem~\ref{thm:weak_weight} and Fact~\ref{fact:cww_iff_each_Q_ww}.

\begin{theorem}\label{theo:gen_oc_vs_weak_wei}
Let $(X,d)$ be an invariant generalised quasi-metric semilattice. Then $(X,d)$ satisfies {\rm(DPC)} if and only if it is componentwisely weakly weighted. 
Moreover, $w$ is a componentwise weak weight for $d$ precisely when $w\approx w_X$, where $w_X$ is defined by $w_X\restriction_{\mathcal Q(x_i)}=w_{x_i}$ for every $i\in I$, where 
$\{x_i\}_{i\in I}$ is a fixed family of representatives of the connected components of $X$.
\end{theorem}

\begin{example}\label{ex:PX_and_LG_cww}
Using Theorem~\ref{theo:gen_oc_vs_weak_wei} we can prove that the generalised quasi-metric meet-semilattices provided in Example~\ref{ex:w_weighted_gqm} are componentwisely weakly weighted. It is in fact easier to see that they satisfy {\rm(DPC)} (see \cite{DGBKTZ}), and so Theorem~\ref{theo:gen_oc_vs_weak_wei} implies that they are componentwisely weakly weighted. 
\begin{compactenum}[(a)] 
\item Let $S$ be a set. For the reasoning above, the generalised quasi-metric semilattice $(\mathcal P(S),d_{\mathcal P(S)})$ is componentwisely weakly weighted. Let us construct a componentwise weak weight. Fix a family $\{A_i\}_{i\in I}$ of representatives of the equivalence classes of $\cong_{d_{\mathcal P(S)}}$. Then define $w\colon \mathcal P(X)\to\R$ as follows: for every $A\subseteq X$, take the index $i\in I$ such that $A\cong_{d_{\mathcal P(S)}}A_i$ and set 
$$w(A)=w_{A_i}(A)=d_{\mathcal P(S)}(A_i,A)-d_{\mathcal P(S)}(A,A_i)=\lvert A\setminus A_i\rvert-\lvert A_i\setminus A\rvert.$$
\item Let us now explicitly provide the componentwise weak weight for the generalised quasi-metric semilattice $(L(G),d_{L(G)})$, where $G$ is an abelian group. Let $\{K_i\}_{i\in I}$ be a family of representatives of the equivalence classes of $\cong_{d_{L(G)}}$. Then define $w\colon L(G)\to\R$ as follows: for every $i\in I$, $w$ coincides with $w_{K_i}$ on $\mathcal Q(K_i)$, i.e., for every subgroup $H$ of $G$ such that $H\cong_{d_{L(G)}}K_i$, 
$$w(H)=w_{K_i}(H)=d_{L(G)}(K_i,H)-d_{L(G)}(H,K_i)=\log\lvert H+K_i:K_i\rvert-\log\lvert H+K_i:H\rvert.$$  
\end{compactenum}
\end{example}

Keeping Theorem~\ref{theo:gen_oc_vs_weak_wei} 
in mind, we generalise the notion of semivaluations and semi-co-valuations  in order to extend Corollary~\ref{coro:w_wei_to_valu2} to generalised quasi-metrics (see Corollary~\ref{coro:w_wei_to_valu2gen}).

\begin{definition}\label{def:gen_sv_and_scv}
Let $(X,\leq)$ be a meet-semilattice. Given a congruence $\cong$ on $X$, a function $f\colon X\to\R$ is called:
\begin{compactenum}[-]
	\item {\em generalised meet valuation with respect to $\cong$} if $f(x)+f(x\wedge y\wedge z)\geq f(x\wedge y)+ f(x\wedge z)$, for every $x,y,z\in X$ satisfying $x\cong x\wedge z$ and $y\cong x\wedge y$;
	\item {\em generalised meet co-valuation with respect to $\cong$} if $f(x)+f(x\wedge y\wedge z)\leq f(x\wedge y)+ f(x\wedge z)$, for every $x,y,z\in X$ satisfying $x\cong x\wedge z$ and $y\cong x\wedge y$.
\end{compactenum}
Let $(X,\leq)$ be a join-semilattice. Given a congruence $\cong$, a function $f\colon X\to\R$ is called:
\begin{compactenum}[-]
	\item {\em generalised join valuation with respect to 
		$\cong$} if $f(x)+f(x\vee y\vee z)\leq f(x\vee y)+f(x\vee z)$, for every $x,y,z\in X$ satisfying $x\cong x\vee z$ and $y\cong x\vee y$;
	\item {\em generalised join co-valuation with respect to 
		$\cong$} if $f(x)+f(x\vee y\vee z)\geq f(x\vee y)+f(x\vee z)$, for every $x,y,z\in X$ satisfying $x\cong x\vee z$ and $y\cong x\vee y$.
\end{compactenum}
Similarly to the non-generalised case, a {\em generalised semivaluation} is either a generalised join valuation or a generalised meet valuation, and a {\em generalised semi-co-valuation} is either a generalised join co-valuation or a generalised meet co-valuation.
\end{definition}

It is easy to see that the notions provided in Definition~\ref{def:gen_sv_and_scv} extend the ones given in Definition~\ref{def:val_and_coval}. In fact, for example, if $X$ is a meet-semilattice, then $f\colon X\to\R$ is a meet co-valuation if and only if it is a generalised meet co-valuation with respect to the trivial congruence having just one equivalence class $X$.

Let us provide an immediate, but useful fact.
\begin{fact}\label{fact:equivalence_classes}
Let $x,y,z$ be three points in a semiliattice $X$ and $\cong$ be a congruence. If $x\cong x\ast z$ and $y\cong x\ast y$, then
$$y\cong x\ast y\cong x\ast y\ast z\cong y\ast z.$$
\end{fact}
\begin{proof}
Since $x\cong x\ast z$ and $\cong$ is a congruence, $x\ast y\cong x\ast y\ast z$. Finally, the convexity of the equivalence class implies that $y\cong y\ast z$ since $y\cong x\ast y\ast z$. 
\end{proof}

The following diagram represents the statement of Fact~\ref{fact:equivalence_classes} if $X$ is a meet-semilattice. The hypotheses are marked by the symbol $\cong$ on the edge connecting the two elements. Then the elements of the equivalence class that the fact relates to  one another 
are framed.
$$
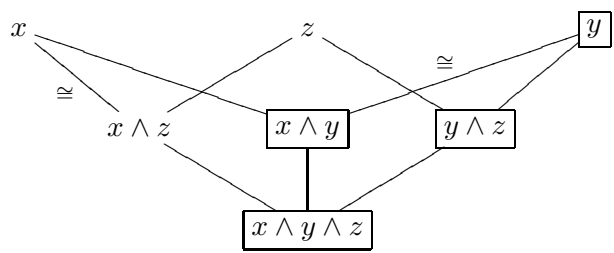
\xymatrix{
x\ar@{-}_{\cong}[dr]\ar@{-}[drr] & & z\ar@{-}[dr]\ar@{-}[dl] & & *+[F]{y}\ar@{-}[dl]\ar@{-}_{\cong}[dll]\\
& x\wedge z\ar@{-}[dr] & *+[F]{x\wedge y} \ar@{-}[d] & *+[F]{y\wedge z} \ar@{-}[dl] &\\
& & *+[F]{x\wedge y\wedge z} & &
}
$$

\begin{proposition}\label{prop:mono_semival+gen_new}
Let $(X,\leq)$ be a semilattice, $\cong$ be a congruence on $X$ and $f\colon X\to \R$ be a function.
\begin{compactenum}[(a)]
	\item The function $f$ is a generalised semivaluation if and only if $-f$ is a generalised semi-co-valuation.
	\item If $f$ is a generalised semi(-co-)valuation, then every $f'\approx f$ with respect to $\{[x]_{\cong}\}_{x\in X}$ has the same property. 
	\item If $f$ is a generalised semivaluation, then $f\restriction_{[x]_{\cong}}$ is a semivaluation for every $x\in X$. In particular, $f\restriction_{[x]_{\cong}}$ is non-decreasing.
	\item If $f$ is a generalised semi-co-valuation, then $f\restriction_{[x]_{\cong}}$ is a semi-co-valuation. In particular, $f\restriction_{[x]_{\cong}}$ is non-increasing.
\end{compactenum}
\end{proposition}

It is important to note that we cannot extend the final implications of items (c) and (d) of Proposition~\ref{prop:mono_semival+gen_new} and state that a generalised semi(-co)-valuation is monotone with respect to the whole semilattice. On an intuitive level, item (b) of Proposition~\ref{prop:mono_semival+gen_new} represents an obstruction to that property. We provide an explicit example of a strictly increasing generalised semi-co-valuation.

\begin{example}
Consider the lattice $\SSS=\{0,1\}$ as defined in Example~\ref{ex:w_weighted_qm}(a), and the congruence $\cong$ given by the equality. It can be verified that the map $f=id_{\mathbb S}$ is a generalised meet co-valuation even though it is strictly increasing.
\end{example}

\begin{corollary}\label{wgensemicoval}
	Let $(X,d)$ be an invariant generalised quasi-metric semilattice and $w$ be a componentwise weak weight for $d$. Then $w$ is a generalised semi-co-valuation with respect to $\cong_d$ such that, for every $x\in X$, $w\restriction_{\mathcal Q(x)}$ is strictly decreasing. Thus $-w$ is a generalised semivaluation with respect to $\cong_d$ such that, for every $x\in X$, $w\restriction_{\mathcal Q(x)}$ is strictly increasing.
\end{corollary}
\begin{proof}
Assume that $(X,\leq_d)$ is a meet-semilattice. Let us take three points $x,y,z\in X$ such that $x\cong_d x\wedge z$ and $y\cong_d x\wedge y$. Then, thanks to Fact~\ref{fact:equivalence_classes} and the invariance of $d$,
\begin{align*}
w(x)&\,=w(x)+d(x,x\wedge z)-d(x,x\wedge z)=w(x\wedge z)-d(x,x\wedge z)\leq\\
&\,\leq w(x\wedge z)-d(x\wedge y,x\wedge y\wedge z)-w(x\wedge y)+w(x\wedge y)=\\
&\,=w(x\wedge z)-w(x\wedge y\wedge z)+w(x\wedge y).
\end{align*}
Hence, $w$ is a generalised meet co-valuation with respect to $\cong_d$.  
Theorem~\ref{theo:w_wei_to_valu_join_1} implies the second part of the statement.

The second statement follows from the first one by applying Proposition~\ref{prop:mono_semival+gen_new}(a).
\end{proof}

\begin{definition}\label{def:d_fgen}
	If $X$ is a meet-semilattice, $\cong$ a congruence on $X$ and $f\colon X\to \R$ is a generalised meet co-valuation, we define $d_{f}\colon X\times X\to\R_{\geq 0}\cup\{\infty\}$ as follows: for every $x,y\in X$,
	$$d_{f}(x,y):=\begin{cases}\begin{aligned}&f(x\wedge y)-f(x) &\text{if $x\cong x\wedge y$,}\\
			&\infty &\text{otherwise.}\end{aligned}\end{cases}$$
	%
	
	If $X$ is a join-semilattice, $\cong$ a congruence on $X$ and
	$f\colon X\to \R$ is a generalised join 
	co-valuation, we define $d_{f}\colon X\times X\to\R_{\geq 0}\cup\{\infty\}$ as follows: for every $x,y\in X$, 
	$$d_{f}(x,y):=\begin{cases}
		\begin{aligned}& f(y)-f(x\vee y)&\text{if $y\cong x\vee y$,}\\
			&\infty &\text{otherwise.}\end{aligned}\end{cases}$$
\end{definition}

\begin{theorem}\label{theo:w_wei_to_valu_2gen}
	Let $(X,\leq)$ be a semilattice and $\cong$ a congruence on it. 
	
	If $f$ is a generalised semi-co-valuation such that $f\restriction_{[x]_{\cong}}$ is strictly decreasing for every $x\in X$, then $d_f$ is an invariant generalised quasi-metric on $X$ satisfying $\leq=\leq_{d_f}$, and it is componentwisely weakly weighted by $f$. Moreover, 
	$\cong\,=\,\cong_{d_f}$.
	
	If $f$ is a generalised semivaluation such that $f\restriction_{[x]_{\cong}}$ is strictly increasing for every $x\in X$, then $d_{-f}$ is an invariant generalised quasi-metric on $X$ satisfying $\leq=\leq_{d_{-f}}$, and it is componentwisely weakly weighted by $-f$. Moreover, 
	$\cong\,=\,\cong_{d_{-f}}$.
\end{theorem}
\begin{proof}
We give a proof when $(X,\leq)$ is a meet-semilattice and $f$ is a generalised meet co-valuation. 

First of all, we claim that, $\cong\, =\,\cong_{d_f}$ (note that $\cong_{d_f}$ can be defined even if $d_f$ is not a priori a generalised quasi-metric). Indeed, let $x,y\in X$. If $x\cong y$, then, since $\cong$ is a congruence, $x\cong x\wedge y\cong y$, and thus $d_f(x,y)<\infty$ and $d_f(y,x)<\infty$. Vice versa, if $d_f(x,y)<\infty$ and $d_f(y,x)<\infty$, then $x\cong x\wedge y$ and $y\cong x\wedge y$, and so $x\cong y$.

Let us now show that $d_f$ is a generalised quasi-metric. Since $f$ is non-increasing on each $[x]_{\cong}$, $d_f$ assumes non-negative values. To prove (QM1), take two points $x,y\in X$ such that $d_f(x,y)=d_f(y,x)=0$. Then the definition of $d_f$ implies that $x\cong x\wedge y\cong y$ and $f(x)=f(x\wedge y)=f(y)$. As $f$ is strictly monotone on $[x]_{\cong}$, the equalities $x=x\wedge y=y$ descend.	
	
Let us prove that $d_f$ satisfies (QM2). For every triple of points $x,y,z\in X$, we claim that
\begin{equation}\label{eq:d_f_QM2}
d_f(y,z)\leq d_f(y,x)+d_f(x,z).
\end{equation}
If either $d_f(y,x)=\infty$ or $d_f(x,z)=\infty$, there is nothing to prove. Therefore, we can assume that $y\cong x\wedge y$ and $x\cong x\wedge z$. Then Fact~\ref{fact:equivalence_classes} implies that also $y\cong y\wedge z$. Since $f$ is a generalised meet co-valuation and $f$ is non-increasing on the equivalence classes of $\cong$, we obtain the following chain:
\begin{align*}f(x\wedge y)-f(y)+f(x\wedge z)-f(x)&\,\geq f(x\wedge y)-f(y)+f(x\wedge z)+f(x\wedge y\wedge z)-f(x\wedge y)-f(x\wedge z)=\\
&\,=f(x\wedge y\wedge z)-f(y)\geq f(y\wedge z)-f(y).\end{align*}
Hence, \eqref{eq:d_f_QM2} descends by plugging in the definition of $d_f$.

Note that, for $x,y\in X$, $x\leq_{d_f}y$ if and only if $x\cong x\wedge y$ and $f(x)=f(x\wedge y)$, which is then equivalent to $x=x\wedge y$ and therefore to $x\leq y$. Hence, $\leq=\leq_{d_f}$. 
Finally, $d_f$ is invariant by construction and it can be easily seen that $f$ is a componentwise weak weight for $d_f$.

The second statement can be derived from the first one by using Proposition~\ref{prop:mono_semival+gen_new}.
\end{proof}

Using the steps provided in this subsection, we can now extend Corollary~\ref{coro:w_wei_to_valu2} to generalised quasi-metrics.

\begin{corollary}\label{coro:w_wei_to_valu2gen}
Let $(X,\leq)$ be a semilattice and $\cong$ a congruence on it. There are one-to-one correspondences between the three sets consisting of the following objects:
	\begin{compactenum}[(a)]
		\item invariant componentwise weakly weighted generalised quasi-metrics $d$ on $X$ satisfying $\leq_d=\leq$;
		\item equivalence classes of generalised semi-co-valuations $f$ strictly decreasing on each $[x]_{\cong}$;
		\item equivalence classes of generalised semivaluations $g$ strictly increasing on each $[x]_{\cong}$.
	\end{compactenum}
	More explicitly, to $d$ we associate $[w]_\approx$ and $[-w]_\approx$ where $w$ is a componentwise weak weight for $d$ and  $\approx$ is relative to $\{\mathcal Q_{d}(x)\}_{x\in X}$. To $[f]_{\approx}$, where $\approx$ is relative to $\{[x]_{\cong}\}_{x\in X}$, we associate $d_f$ and $[-f]_\approx$. Finally, to $[g]_{\approx}$, where $\approx$ is relative to $\{[x]_{\cong}\}_{x\in X}$, we associate $d_{-g}$ and $[-g]_\approx$.
%
%
%
\end{corollary}
\begin{proof} 
Let $(X,\leq)$ be a meet-semilattice with a congruence $\cong$.
We prove 	the correspondence between (a) and (b) since that between (b) and (c) is trivial by Theorem~\ref{theo:w_wei_to_valu_2gen} and Proposition~\ref{prop:mono_semival+gen_new}(a).

Both associations described in the statement are well-defined by virtue of 
Corollaries~\ref{approx} and~\ref{wgensemicoval}, Proposition~\ref{prop:mono_semival+gen_new}(b) and Theorem~\ref{theo:w_wei_to_valu_2gen} since two generalised semi-co-valuations $f,f'$ of $X$ satisfy $d_f=d_{f'}$ provided that $f\approx f'$. 

Moreover, we have to verify that the two associating maps are one the inverse of the other. Indeed, given a generalised semi-co-valuation $f$ of $X$ strictly decreasing on each $[x]_{\cong}$, we have seen in Theorem~\ref{theo:w_wei_to_valu_2gen} that $f$ is a generalised weak weight for $d_f$; and then to $d_f$ we associate $[f]_\approx$ as in Corollary~\ref{wgensemicoval}. 
On the other hand, to check that if $d$ is an invariant generalised quasi-metric on $X$ with $\leq_d=\leq$ and componentwise weakly weighted by $w$, then $d=d_w$, proceed as in  the proof of Corollary~\ref{coro:w_wei_to_valu2}, taking into account that, for $x,y\in X$, $d(x,y)=\infty$ if and only if $d_w(x,y)=\infty$ as $\cong\,=\,\cong_d$ by Theorem~\ref{theo:w_wei_to_valu_2gen}. 
\end{proof}

\section{Entropy for meet-semilattices with respect to an equivalence relation}\label{entropysec}

A self-map $\phi\colon X\to X$ of a meet-semilattice $X$ is an {\em endomorphism} if, for every $x,y\in X$, $\phi(x\wedge y)=\phi(x)\wedge\phi(y)$. 
For $n\in\N_+$, the \emph{$n$-th trajectory} of $\phi$ in $x\in X$ is
	$$T_{n}(\phi,x):=x\wedge\phi(x)\wedge\cdots\wedge\phi^{n-1}(x).$$
Fix a function $f\colon X\to \R$ and, for $x\in X$, let
\begin{equation}\label{newheq} 
h_f(\phi,x):=\limsup_{n\to \infty}\frac{f(T_{n}(\phi,x))}{n}.
\end{equation}
Clearly, if the function $g\colon X\to\R$ satisfies $g\sim f$, then $h_g(\phi,x)=h_f(\phi,x)$ for every $x\in X$.

\smallskip
We introduce in the following definition a new notion of intrinsic entropy for endomorphisms of meet-semilattices with respect to  an equivalence relation.

\begin{definition}\label{newh}
Let $X$ be a meet-semilattice, $\cong$ an equivalence relation on $X$ and $\phi$ an endomorphism of $X$.
A point $x\in X$ is \emph{$(\phi,{\cong})$-inert} if $T_{n}(\phi,x)\cong x$ 
for every $n\in\N_+$.

For a function $f\colon X\to \R$, the {\em intrinsic semilattice entropy relative to $\cong$} is defined as
$$\widetilde h_{f,\cong}(\phi):=\sup\{h_f(\phi,x)\mid x\in X,\ x\ \text{is $(\phi,{\cong})$-inert}\}.$$
\end{definition}

Whenever the equivalence relation $\cong$ is a congruence on $X$, we can provide the following characterisations of  a $(\phi,\cong)$-inert point. We say that an endomorphism $\phi$ of a meet-semilattice $X$ {\em respects} an equivalence relation $\cong$ if $\phi(x)\cong\phi(y)$ for every $x,y\in X$ satisfying $x\cong y$.

\begin{lemma}\label{rem:newTQ} 
Let $X$ be a meet-semilattice, $\cong$ a congruence on $X$, $\phi$ an endomorphism of $X$ and $x\in X$.
Then the following conditions are equivalent:
\begin{compactenum}[(a)]
\item $x$ is $(\phi,{\cong})$-inert;
\item  $x\wedge\phi^n(x)\cong x$ for every $n\in\N$.
\end{compactenum}
Moreover, they imply the following item:
\begin{compactenum}[(c)]
	\item $x\wedge\phi(x)\cong x$.
\end{compactenum}
\medskip
\noindent
In addition, if $\phi$ respects $\cong$, then (a) and (b) are equivalent to (c).
\end{lemma}
\begin{proof}
(a)$\Rightarrow$(b) For every $n\in\N_+$ one has that $x\cong T_{n+1}(\phi,x)$  by assumption. By the convexity of $[x]_{\cong}$ provided by Fact~\ref{prop:congruence}, also $x\wedge\phi^{n}(x)\cong x$.

(b)$\Rightarrow$(a) For $n=1$ one has $T_1(\phi,x)=x\cong x$. Thus, we proceed by induction and assume that $T_n(\phi,x)\cong x$ for $n\in\mathbb N_+$. Notice that $T_{n+1}(\phi,x)=\phi^n(x)\wedge T_n(\phi,x)$. Thus, since $\cong$ is a congruence, (b) and the induction hypothesis imply that
$T_{n+1}(\phi,x)\cong\phi^n\wedge x\cong x,$
which concludes the proof.

The implication (b)$\Rightarrow$(c) is trivial.

(c)$\Rightarrow$(a) Assume that $x\cong x\wedge \phi(x)$. As $T_1(\phi,x)=x\cong x$, we proceed by induction: if $T_{n}(\phi,x)\cong x$ for some $n\in\mathbb N_+$, then $\phi(T_{n}(\phi,x))\cong\phi(x)$, and so, since $\cong$ is a congruence, we have
\[T_{n+1}(\phi,x)=x\wedge \phi(T_{n}(\phi,x))\cong x\wedge\phi(x)\cong x.\qedhere\]
\end{proof}

 Note that item (c) of Lemma \ref{rem:newTQ} is strictly weaker than (a) and (b) if $\phi$ does not respect $\cong$. We provide a small counterexample.
 
\begin{example}
Consider the lattice $X$ described by the following Hasse diagram
$$
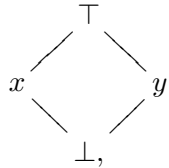
\xymatrix@-1pc{& \top\ar@{-}[dr]\ar@{-}[dl] &\\
	x\ar@{-}[dr] & & y\ar@{-}[dl]\\
	& \perp, &
}$$
the endomorphism $\phi$ defined by $\phi(\top)=x$, $\phi(x)=\phi(y)=\phi(\perp)=\perp$, and the congruence $\cong$ partitioning $X$ into the two subsets $\{\top,x\}$ and $\{y,\perp\}$. Note that $\phi$ does not respect $\cong$ since $\phi(\top)\not\cong\phi(x)$ even though $\top\cong x$. The point $\top$ satisfies item (c) of Lemma \ref{rem:newTQ}, but $\top\wedge\phi^2(\top)=\perp\not\cong\top$.
\end{example}

\begin{remark}
Let $X$ be a meet-semilattice together with a congruence $\cong$, and let $\phi$ be an endomorphism of $X$ that respects $\cong$.
\begin{compactenum}[(a)]
\item If $x\in X$ is $(\phi,\cong)$-inert, then every $y\in [x]_{\cong}$ is $(\phi,\cong)$-inert. Indeed, $y\cong x\cong x\wedge\phi(x)\cong y\wedge\phi(y)$ and then Lemma \ref{rem:newTQ}(c) applies.
\item Consequently, the set $\mathcal I_\phi(X,\cong)$ of all $(\phi,\cong)$-inert points of $X$ is a disjoint union of congruence classes of $X$.
\end{compactenum}
\end{remark}

\subsection{Entropy for generalised normed semilattices}\label{hSesec}

Following~\cite{uatc,DGBKTZ}, a \emph{generalised normed meet-semilattice} is a meet-semilattice $X$ with a function $v\colon X\to\R_{\geq 0}\cup\{\infty\}$, and so $v$ is called \emph{generalised norm}. 
Let  $\cong_v$ be the equivalence relation induced by the following partition: $$\mathcal F_v(X):=\{x\in X\mid v(x)<\infty\}\quad \text{and}\quad X_\infty:=X\setminus\mathcal F_v(S).$$
 Define $f_v\colon X\to \R_{\geq0}$ by letting $f_v(x)=v(x)$ if $x\in\mathcal F_v(X)$ and $f_v(x)=0$ otherwise.

The following result shows that $h_{f_v,\cong_v}(\phi)$ coincides with the \emph{semigroup entropy} $h_v(\phi)$ defined in~\cite{DGBKTZ} 
by slightly generalising that invented 
in~\cite{uatc}.
It is worth to recall that many of the classical entropy functions in mathematics, such as the measure entropy~(\cite{K,Sinai}), several instances of topological entropy~(\cite{AKM,Hood,G}), many kinds of (adjoint) algebraic entropy~(\cite{AKM,DG,DGS,DGSZ,GK,SZ,W}) and the two set-theoretic entropies~(\cite{AZD,DG-islam}), are found to be semigroup entropies, and in 
almost all these cases the underlying semigroup is a semilattice.

\begin{proposition}\label{hv} 
Let $v\colon X\to\R_{\geq 0}\cup\{\infty\}$ be a generalised norm on a meet-semilattice $X$ that is subadditive on $\mathcal F_v(S)$, and let $\phi\colon X\to X$ be a monotone endomorphism. Then, 
for every $x\in \mathcal F_v(X)$, $x$ 
is $(\phi,\cong_v)$-inert and
\begin{equation}\label{hveq}
h_{f_v}(\phi,x)=
	\limsup_{n\to \infty}\frac{v(T_{n}(\phi,x))}{n}.
		\end{equation}
In particular, $\widetilde h_{f_v,\cong_v}(\phi)=\sup\{h_{f_v}(\phi,x)\mid x\in\mathcal F_v(X)\}$ and so $\tilde h_{f_v,\cong_v}(\phi)=h_v(\phi)$.
\end{proposition}
\begin{proof} 
Let  $x\in\mathcal F_v(X)$. Since $\phi$ monotone, for all $n\in\N_+$, $v(\phi^n(x))\leq v(x)<\infty$.
As  $v$ is subadditive on $\mathcal F_v(X)$, one computes
$$v(T_1(\phi,x))=v(x\wedge \phi(x))\leq v(x)+v(\phi(x))\leq 2v(x)<\infty$$
and, by induction on $n\in\N_+$, one has
$$v(T_{n+1}(\phi,x))\leq v(T_n(\phi,x))+v(\phi^n(x))\leq v(T_n(\phi,x))+v(x)<\infty;$$
 hence $x$ is $(\phi,\cong_v)$-inert.
 
The last assertion follows since every $x\in \mathcal F_v(X)$ is $(\phi,\cong_v)$-inert and $h_{f_v}(\phi,x)=0$ for every $(\phi,\cong_v)$-inert 
$x\in X_\infty$, and by the definition of semigroup entropy.
\end{proof}

The limit superior in \eqref{hveq} is actually a limit as proved in~\cite{uatc}.

\begin{remark}
Let $v\colon X\to\R_{\geq 0}\cup\{\infty\}$ be a generalised norm on a meet-semilattice $X$, and let $\phi$ be an endomorphism of $X$.

If $v$ is non-decreasing, then all the elements of $\mathcal F_v(X)$ are $(\phi,\cong_v)$-inert and the conclusions of Proposition~\ref{hv} hold true.

On the other hand, in case $v$ is non-increasing, all the points of $X_\infty$ are $(\phi,\cong_v)$-inert. 
Therefore, if $v$ is also subadditive on $\mathcal F_v(S)$ and $\phi$ is monotone, as in Proposition~\ref{hv} every $x\in X$ 
is $(\phi,\cong_v)$-inert and
\begin{equation*}
h_{f_v}(\phi,x)=\begin{cases}\begin{aligned} &\limsup_{n\to \infty}\frac{v(T_{n}(\phi,x))}{n}&\text{if $x\in\mathcal F_v(X)$,}\\
			&0 &\text{otherwise.}\end{aligned}\end{cases}\end{equation*}
\end{remark}

\begin{remark}
Let $v\colon X\to\R_{\geq 0}\cup\{\infty\}$ be a generalised norm on a meet-semilattice $X$, and let $\phi$ be an endomorphism of $X$.  If $v$ is non-increasing and subadditive, then $\cong_v$ is a congruence: 
for every $x,y\in X$, $x\wedge y\in\mathcal F_v(X)$ if and only if both $x$ and $y$ belong to $\mathcal F_v(X)$. Moreover, if $\phi$ is monotone, then $\phi(\mathcal F_v(X))\subseteq \mathcal F_v(X)$. However, this is not enough to get that $\phi$ respects $\cong_v$: one needs either $\phi(X_\infty)\subseteq X_\infty$ or $\phi(X_\infty)\subseteq\mathcal F_v(X)$.
\end{remark}

\subsection{Entropy for componentwisely weakly weighted invariant  generalised quasi-metric semilattices}\label{tildehsec}

Let $(X,d)$ be an invariant generalised quasi-metric meet-semilattice together with a componentwise weak weight $w$ for $d$. Let $\phi$ be a non-expansive endomorphism of $X$ (i.e., $d(\phi(x),\phi(y))\leq d(x,y)$). Hence, we have the following data:
\begin{compactenum}[-]
\item  the equivalence relation $\cong_d$ is a congruence and so every connected component $\mathcal Q(x)$ is a convex invariant subsemilattice of $X$ (see Proposition~\ref{congd});
\item $\phi$ respects $\cong_d$ by definition of non-expansive endomorphism;
\item for every $x\in X$ one has that $w\restriction_{\mathcal Q(x)}\,\sim w_{x}$, where
$w_{x}=d(x,y)-d(y,x)<\infty$ for every $y\in \mathcal Q(x)$ (see Theorem~\ref{theo:gen_oc_vs_weak_wei} and 
\eqref{eq:w_x}).
\end{compactenum}

\medskip
The following result shows in particular that $\widetilde h_{w,\cong_d}(\phi)$ coincides with the \emph{intrinsic semilattice entropy} $h_d(\phi)$ from~\cite{CDFGBT}. This entropy was introduced to get a unifying approach able to cover all (or at least, most) of the intrinsic-like entropies in literature:
 the intrinsic algebraic entropy~(\cite{DGSV}), the algebraic and the topological entropies for locally linearly compact vector spaces~(\cite{CGBalg,CGBtop}), the topological entropy for totally disconnected locally compact abelian groups~(\cite{DGB-BT,GBV}) and the algebraic entropy for compactly covered locally compact abelian groups~(\cite{DGB-BT,GBST}).

\begin{proposition}\label{prop:cww int} 
Let $(X,d)$ be an invariant generalised quasi-metric meet-semilattice componentwisely weakly weighted by $w$ and $\phi$ a non-expansive endomorphism of $X$. 
\begin{compactenum}[(a)]
\item An element $x\in X$ is $(\phi,\cong_d)$-inert if and only if $d(x,\phi(x))<\infty$.
\item For every $x\in\mathcal I_{\phi}(X,\cong_d)$,
\begin{equation}\label{hdeq}h_{w}(\phi,x)=\limsup_{n\to\infty}\frac{d(x,T_{n}(\phi,x))}{n}.\end{equation}
\end{compactenum}
In particular, $h_{w}(\phi,x)$ does not depend on the choice of $w$.

Moreover, $\widetilde h_{w,\cong_d}(\phi)=\sup\{\widetilde h_{w}(\phi,x)\mid \text{$x\in X$ such that $d(x,\phi(x))<\infty$}\}$  and so $\widetilde h_{w,\cong_d}(\phi)=\widetilde h_d(\phi)$.
\end{proposition}
\begin{proof} 
(a) For every $x\in X$, $d(x,\phi(x))=d(x,x\wedge \phi(x))$ by the invariance of $d$. Moreover, as $x\wedge\phi(x)\leq x$, $d(x\wedge\phi(x),x)=0$ and so $x\wedge\phi(x)\cong_d x$ if and only if $d(x,\phi(x))<\infty$. Lemma~\ref{rem:newTQ} concludes the proof.

(b) Let $x\in\mathcal I_{\phi}(X,\cong_d)$. First notice that, for every $n\in\N_+$, $T_n(\phi,x)\in[x]_{\cong_d}$. Since $w\restriction_{\mathcal Q(x)}\,\sim w_{x}$ as stated above, there exists a constant $c_x>0$ such that $w(T_n(\phi,x))=w_x(T_n(\phi,x))+c_x$ for every $n\in\N_+$. Moreover, since $T_n(\phi,x)\leq x$, $w_x(T_n(\phi,x))=d(x,T_n(\phi,x))$, and so  one computes
\[h_{w}(\phi,x)=\limsup_{n\to\infty}\frac{w(T_{n}(\phi,x))}{n}=\limsup_{n\to\infty}\frac{w_x(T_{n}(\phi,x))}{n}=\limsup_{n\to\infty}\frac{d(x,T_{n}(\phi,x))}{n}.\]

The last assertion follows from the definitions.
\end{proof}

The limit superior in \eqref{hdeq} is actually a limit, as proved in \cite[Theorem~3.15]{CDFGBT}.

\medskip
The requirement that the generalised quasi-metric meet-semilattice has to satisfy (DPC) may seem restrictive. However, this property is fulfilled by all the examples of intrinsic entropies collected in~\cite{CDFGBT}. Nevertheless, we explicitly pose the following problem. 

\begin{question}\label{q:non_DPC}
Let $\phi$ be a non-expansive endomorphism of an invariant generalised quasi-metric meet-semilattice $(X,d)$. Fix a family $\{x_i\}_{i\in I}$ of representatives of the connected components of $X$, and define $w_X$ as in Theorem \ref{theo:gen_oc_vs_weak_wei}, i.e., by setting $w_X\restriction_{\mathcal Q(x_i)}=w_{x_i}$, for every $i\in I$. For a $(\phi,\cong_d)$-inert point $x\in X$, is it true that the value $\widetilde h_{w_X}(\phi,x)$ does not depend on the choice of the representative family? If so, can we conclude that $h_{w_X}(\phi,x)=\limsup_{n\to\infty}\frac{d(x,T_{n}(\phi,x))}{n}$?
\end{question}

%
%

\subsection{How to recover some classical entropies}
We conclude by listing some concrete and classical examples of entropies obtained using Example~\ref{ex:w_weighted_gqm} (see also Examples~\ref{ex:qm groups} and~\ref{ex:PX_and_LG_cww}).

\subsubsection{The algebraic entropy}\label{algss}

Let $G$ be an abelian group and $(L(G),\supseteq)$ the meet-semilattice consisting of all subgroups of $G$ ordered by containment. 
Every group endomorphism $\phi\colon G\to G$ induces an endomorphism of $L(G)$, which we call again $\phi\colon L(G)\to L(G)$, by setting $H\mapsto \phi(H)$ for every $H\in L(G)$.
For $H\in L(G)$, $$T_n(\phi,H)=H+\phi(H)+\cdots+\phi^{n-1}(H),\quad \text{for every $n\in\N_+$}.$$

\begin{compactenum}[(a)]
\item The meet-semilattice $(L(G),\supseteq)$ can be equipped with the subadditive generalised norm $$v\colon L(G)\to \R\cup\{\infty\},\quad v(H):=\log|H|.$$ In particular, $\phi\colon L(G)\to L(G)$ is monotone with respect to $v$. By Proposition~\ref{hv}, one has 
$$\widetilde h_{f_v,\cong_v}(\phi)=\sup\{h_{f_v}(\phi,H)\mid H\in\mathcal F_v(L(G))\}.$$
Since $\mathcal F_v(L(G))=\{H\leq G\mid H\text{ finite}\}$ and $f_v\colon L(G)\to \R$  is given by $f_v(H)=\log|H|$ if $H$ is finite and $f_v(H)=0$ otherwise,
we get that 
\begin{equation}\label{enteq}
h_{f_v}(\phi,H)=\lim_{n\to\infty}\frac{\log|T_n(\phi,H)|}{n}, 
\quad\text{and so}\quad \tilde h_{f_v,\cong_v}(\phi)=\ent(\phi),
\end{equation}
where $\ent(\phi)$ is the classical \emph{algebraic entropy} of  abelian group endomorphisms (\cite{AKM,uatc,DGSZ,W}). 

\item The pair $(L(G),d_{L(G)})$ is an invariant generalised quasi-metric meet-semilattice satisfying {\rm (DPC)}. 
Suppose that $w$ is a componentwise weak weight for $d_{L(G)}$. By Proposition~\ref{prop:cww int}, we obtain that, for every $(\phi,\cong_{d_{L(G)}})$-inert subgroup $H$ of $G$,
$$h_{w}(\phi,H)=\lim_{n\to\infty}\frac{\log\lvert T_{n}(\phi,H):H\vert}{n}, 
\quad \text{and so}\quad \widetilde h_{w}(\phi)=\widetilde{\ent}(\phi),$$
where $\widetilde{\ent}(\phi)$ is the \emph{intrinsic algebraic entropy} of $\phi$ as defined in~\cite{DGSV} (see also~\cite[\S5.2.1]{CDFGBT}).
\end{compactenum}

\begin{remark}\label{rem911}
In the above setting, 
as $\mathcal F_v(L(G))$ is a $\phi$-invariant subsemilattice of $L(G)$, we can consider on it the restrictions $\phi_{fin}$ and $v_{fin}$ of $\phi$ and $v$, respectively. Restricting the generalised quasi-metric $d_{L(G)}$ defined in Example~\ref{ex:w_weighted_gqm}(b), we obtain a quasi-metric $d$ on $\mathcal F_v(L(G))$ for which $v_{fin}$ is a weight. Therefore, $v_{fin}$ is a strictly decreasing meet co-valuation by Theorem~\ref{theo:w_wei_to_valu_join_1}. 
Since $(\mathcal F_v(L(G)),d)$ has only one connected component, for every 
$E\in\mathcal F_v(G)$, $w_E\sim v_{fin}$ by Theorem~\ref{thm:weak_weight}  and Proposition~\ref{prop:w+c}(a), where $w_E$ is defined in Example~\ref{ex:PX_and_LG_cww}(b).
Thus, 
fixed any $E\in\mathcal F_v(G)$, for every $H\in \mathcal F_v(L(G))$, 
$$h_{f_v}(\phi,H)\,=h_{w_E}(\phi_{fin},H)=\lim_{n\to\infty}\frac{\log\lvert T_n(\phi,H)+E:E\rvert-\log\lvert T_n(\phi,H)+E:T_n(\phi,H)\rvert}{n},$$
which has to be compared with \eqref{enteq}.
\end{remark} 

\begin{example}\label{exlca}
Let now $G$ be a compactly covered locally compact abelian group (i.e., the Pontryagin dual $\widehat G$ of $G$ is a totally disconnected locally compact abelian group) and $\phi\colon G\to G$ a continuous endomorphism. Denote by $B(G)$ the subsemilattice of $L(G)$ consisting of all compact open subgroups of $G$ and
call $d_{B(G)}$ the restriction of $d_{L(G)}$ to $B(G)$.

In this case, $(B(G),d_{B(G)})$ is an invariant quasi-metric meet-semilattice satisfying {\rm (DPC)}, that is, we have only one connected component: for every $U,V\in B(G)$, $U\cong_{d_{B(G)}} V$; in fact,  
$|U+V:V|$ and $|U+V:U|$ are finite since $U$ and $V$ are open in the compact subgroup $U+V$.

Hence, fixed $V\in B(G)$, by Theorem~\ref{thm:weak_weight} and Proposition~\ref{prop:cww int}(a), 
we get that $$h_{w_V}(\phi,U)=\limsup_{n\to\infty}\frac{\log\vert T_{n}(\phi,U):U\vert}{n},\quad\text{for every}\ U\in B(G),$$
and so $$h_{w_V}(\phi)=\sup\{h_{w_V}(\phi,U)\mid U\in B(G)\}=h_{alg}(\phi),$$
where $h_{alg}(\phi)$ is the \emph{algebraic entropy} of $\phi$ as described in~\cite[Proposition~2.2]{DGB-BT}

For the general definition of the algebraic entropy for continuous endomorphisms of locally compact groups, see~\cite{DG-islam,V}. 
In~\cite{uatc}, it was already obtained as a semigroup entropy in the general case.
\end{example}

\begin{remark}
One can carry out considerations similar to those of Example~\ref{exlca} also in other settings. For example, they can be applied to the topological entropy of continuous endomorphisms of totally disconnected locally compact groups (see~\cite{CDFGBT,DGB-BT,uatc,GBV}),
and to the algebraic and the topological entropy of continuous endomorphisms of locally linearly compact vector spaces
(see~\cite{CGBalg,CGBtop}).
\end{remark}

\subsubsection{The set-theoretic entropy}

Analogously to \S\ref{algss} for the algebraic entropy, one can treat the set-theoretic entropy as follows.

Let $S$ be a non-empty set and $(\mathcal P(S),\supseteq)$ the meet-semilattice consisting of all subsets of $S$. 
Every self-map $\varphi\colon S\to S$ induces an endomorphism of $\mathcal P(S)$, that we call again $\varphi\colon\mathcal P(S)\to\mathcal P(S)$, such that $A\mapsto \varphi(A)$ for every $A\in\mathcal P(S)$.
Note that $$T_n(\varphi,A)=A\cup \varphi(A)\cup\cdots\cup \varphi^{n-1}(A),\quad \text{for every $n\in\N_+$}.$$

\begin{compactenum}[(a)]
\item The meet-semilattice  $(\mathcal P(S),\supseteq)$ can be equipped with the subadditive generalised norm 
$$v_{set}\colon \mathcal P(S)\to \R\cup\{\infty\},\quad v_{set}(A):=|A|.$$ In particular, $\varphi\colon \mathcal P(S)\to \mathcal P(S)$ is monotone with respect to $v_{set}$. By Proposition~\ref{hv}, 
$$\widetilde h_{f_{v_{set}},\cong_{v_{set}}}(\varphi)=\sup\{h_{f_{v_{set}}}(\varphi,A)\mid A\in\mathcal F_{v_{set}}(\mathcal P(S))\}.$$
Since $\mathcal F_{v_{set}}(\mathcal P(S))=\{A\subseteq S\mid A\ \text{finite}\}$ and $f_{v_{set}}\colon \mathcal P(S)\to \R$  is given by $f_{v_{set}}(A)=|A|$ if $A$ is finite and $f_{v_{set}}(A)=0$ otherwise, we get that 
$$h_{f_{v_{set}}}(\varphi,A)=\lim_{n\to\infty}\frac{|T_n(\varphi,A)|}{n}, 
\quad \text{and so}\quad \tilde h_{f_{v_{set}},\cong_{v_{set}}}(\varphi)=h_{set}(\varphi),$$
 where $h_{set}(\phi)$ is the classical \emph{set-theoretic entropy} of $\varphi$ (\cite{AZD,uatc,DGBKTZ}). 

\item The pair $(\mathcal P(S),d_{\mathcal P(S)})$ is an invariant generalised quasi-metric meet-semilattice satisfying {\rm (DPC)}. 
If $w$ is a componentwise weak weight for $d_{\mathcal P(S)}$, by Proposition~\ref{prop:cww int} one obtains that, for every $(\phi,\cong_{d_{\mathcal P(S)}})$-inert subset $A$ of $S$,
$$h_{w}(\varphi,A)=\lim_{n\to\infty}\frac{\lvert T_{n}(\varphi,A)\backslash A \vert}{n},\quad\text{and so}\quad \tilde h_{w,\cong_{d_{\mathcal P(S)}}}(\varphi)=\widetilde{h}_{set}(\varphi),$$ 
where $\widetilde{h}_{set}(\varphi)$ is the \emph{intrinsic set-theoretic entropy} of $\varphi$ from~\cite{DGBKTZ}. In the same paper, it was also proved that 
always $\widetilde h_{set}(\varphi)=h_{set}(\varphi)$. 
\end{compactenum}

\medskip
Similar considerations as those in Remark~\ref{rem911} can be done also for the set-theoretic entropy.

\Addresses

\end{document}